%% file: ta-bisim-tr.tex
\documentclass[copyright,creativecommons]{eptcs}

\usepackage[utf8]{inputenc}
\usepackage[T1]{fontenc}
\usepackage[caption=false]{subfig}
\usepackage{amssymb}
\usepackage{amsmath}
\usepackage{xspace}
\usepackage{tikz}
\usepackage{etoolbox}
\usepackage{multicol}
\usepackage{booktabs}
\usepackage{adjustbox}
\usepackage{booktabs}
\usepackage{pgfplots}
\usepackage[ruled,noend,linesnumbered,commentsnumbered]{algorithm2e}
\usepackage{amsthm}

\newcommand{\ie}{\mbox{i.e.,}\xspace}
\newcommand{\eg}{\mbox{e.g.,}\xspace}

\newcommand{\wrt}{\mbox{w.r.t.}\xspace}
\newcommand{\Wlog}{\mbox{w.l.o.g.}\xspace}
\newcommand{\etal}{\mbox{et al.}\xspace}
\newcommand{\cf}{\mbox{cf.}\xspace}
\newcommand{\angles}[1]{\ensuremath{\langle#1\rangle}}
\newcommand{\timbrcheck}{\textsc{TimBrCheck}\xspace}
\newcommand{\uppaal}{\textsc{Uppaal}\xspace}

\newcommand{\scalefactor}{.788}

\DeclareMathOperator{\histories}{histories}
\DeclareMathOperator{\join}{join}

\pgfplotsset{compat=1.14}
\usepgfplotslibrary{statistics}

\title{Checking Timed Bisimulation with Bounded Zone-History Graphs -- Technical Report}
\author{Lars Luthmann\thanks{This work was funded by the Hessian LOEWE initiative within the Software-Factory 4.0 project.}
\institute{Real-Time Systems Lab\\TU Darmstadt, Germany}
\email{lars.luthmann@es.tu-darmstadt.de}
\and
Hendrik Göttmann
\institute{Real-Time Systems Lab\\TU Darmstadt, Germany}
\email{hendrik.goettmann@es.tu-darmstadt.de}
\and
Isabelle Bacher
\institute{Real-Time Systems Lab\\TU Darmstadt, Germany}
\email{isabelle.bacher@stud.tu-darmstadt.de}
\and
Malte Lochau$^{*}$
\institute{Model-based Engineering Group\\University of Siegen, Germany}
\email{malte.lochau@es.tu-darmstadt.de}
}
\def\titlerunning{Checking Timed Bisimulation with Bounded Zone-History Graphs -- Technical Report}
\def\authorrunning{L. Luthmann, H. Göttmann, I. Bacher, and M. Lochau}

\providecommand{\publicationstatus}{} 

\usetikzlibrary{decorations.pathmorphing,calc,arrows,positioning,automata,angles}
\tikzset{
	every node/.style={},
	descr/.style={fill=white, inner sep=2pt},
	font=\footnotesize,
	initial text={},
	every initial by arrow/.style={*->,>=stealth'},
	labeledroundedstate/.style={rounded corners,draw,inner sep=.2em,align=center},
	labeledroundedstateempty/.style={rounded corners,draw,inner sep=.2em,align=center,transparent},
	labeledroundedstateTA/.style={rounded corners,draw,text width=4.6em,align=center},
	labeledstate/.style={inner sep=.2em,align=center,align=center},
	labeledstateempty/.style={inner sep=.2em,align=center,transparent},
	transition/.style={->,>=stealth'},
}

\renewcommand{\xrightarrow}[1]{
	\mathrel{
		\!\!
		\tikz[baseline=-\the\dimexpr\fontdimen22\textfont2\relax]{
			\node[anchor=south,font=\scriptsize, inner ysep=1.5pt,outer xsep=2.2pt](x){\ensuremath{\!#1}};
			\draw[shorten <=3.4pt,shorten >=3.4pt,->](x.south west)--(x.south east);
		}
		\!\!
	}
}

\newcommand{\xtwoheadrightarrow}[1]{
	\mathrel{
		\!\!
		\tikz[baseline=-\the\dimexpr\fontdimen22\textfont2\relax]{
			\node[anchor=south,font=\scriptsize, inner ysep=1.5pt,outer xsep=2.2pt](x){\ensuremath{\!#1}};
			\draw[shorten <=3.4pt,shorten >=3.4pt,->>](x.south west)--(x.south east);
		}
		\!\!
	}
}

\newcommand{\xTwoheadrightarrow}[1]{
	\mathrel{
		\begin{tikzpicture}[baseline= {( $ (current bounding box.south) + (0,-0.5ex) $ )}]
  			\node[inner sep=.5ex] (a) {$\!\scriptstyle #1\,$};
  			\path[draw,implies-,double distance between line centers=2pt] (a.south east) -- (a.south west);
  			\path[draw,implies-,double distance between line centers=2pt] ($ (a.south east) + (-0.079,0) $)--($ (a.south east) + (-0.08,0) $);
		\end{tikzpicture}
	}
}

\newcommand{\xrightsquigarrow}[1]{
	\mathrel{
		\begin{tikzpicture}[baseline= {( $ (current bounding box.south) + (0,-0.5ex) $ )}]
			\node[inner sep=.5ex] (a) {\ensuremath{\scriptstyle#1}};
			\path[draw,<-,decorate,decoration={zigzag,amplitude=.9pt,segment length=1.2mm,pre=lineto,pre length=3pt}] (a.south east) -- (a.south west);
		\end{tikzpicture}
	}
}

\newcommand{\Twoheadrightarrow}{\xTwoheadrightarrow{\hspace{.45em}}}

\newtheorem{definition}{Definition}
\newtheorem{theorem}{Theorem}
\newtheorem{lemma}{Lemma}

\newtheorem{proposition}{Propostion}
\newtheorem{example}{Example}

\begin{document}

\maketitle

\input{sections/abstract}
\input{sections/introduction}
\input{sections/background}
\input{sections/timed-bisimulation}
\input{sections/implementation.tex}
\input{sections/evaluation}
\input{sections/conclusion}

\bibliographystyle{eptcs}
\bibliography{sections/references}

\end{document}

%% file: sections/abstract.tex

\begin{abstract}
Timed automata (TA) are a well-established
formalism for specifying discrete-state/continuous-time behavior of
time-critical reactive systems.
Concerning the fundamental analysis problem of comparing a candidate implementation 
against a specification, both given as TA,
it has been shown that timed trace equivalence is undecidable,
whereas timed bisimulation equivalence is decidable.
The corresponding proof utilizes region graphs, a
finite, but generally very space-consuming characterization of TA semantics.
Hence, most practical TA tools utilize zone graphs instead, 
a symbolic and generally more efficient representation of TA semantics,
to automate analysis tasks.
However, zone graphs only produce sound results for 
analysis tasks being reducible to plain reachability problems thus
being too imprecise for checking timed bisimilarity.
In this paper, we propose bounded zone-history graphs, 
a novel characterization of TA semantics facilitating an adjustable trade-off 
between precision and scalability of timed-bisimilarity checking.
Our tool \timbrcheck{} is, to the best of our knowledge, the only currently available tool for 
effectively checking timed bisimilarity and even supports non-deterministic TA with silent moves.
We further present experimental results gained from applying our tool to
a collection of community benchmarks, providing insights into trade-offs 
between precision and efficiency, depending on the bound value.
\end{abstract}

%% file: sections/introduction.tex
%
\section{Introduction}\label{section:introduction}

\paragraph{Background and Motivation.}
\emph{Timed automata (TA)} are frequently 
used to specify discrete-state/\allowbreak{}continuous-time behavior of
time-critical reactive (software) systems~\cite{Alur1990,Bengtsson2003}.
TA therefore extend labeled state-transition graphs of classical automata models 
by a set $C$ of \emph{clocks} constituting constantly and 
synchronously increasing, yet independently resettable 
numerical read-only variables.
Clock values are referenced within \emph{clock constraints} in order
to specify boundaries for time intervals to be satisfied 
by occurrences of actions in valid runs of a TA model.

A fundamental analysis problem arises from 
the comparison of a candidate implementation against a specification, both given as TA.
It has been shown that \emph{timed trace inclusion} is undecidable,
whereas \emph{timed (bi-)simulation} is decidable thus making timed bisimilarity a particularly
useful equivalence notion for verifying time-critical behaviors~\cite{Cerans1992,Waez2013}.
The original proof is based on \emph{region graphs}, 
a finite, but generally very space-consuming representation of TA semantics
(\ie{} having $\mathcal{O}(|C|!\cdot k^{|C|})$ many regions, where
$k$ is the maximum constant occurring in a clock constraint).
Instead, most recent practical TA analysis tools use \emph{zone graphs}, constituting
a symbolic and, on average, more efficient representation of TA semantics
as compared to region graphs.
However, zone graphs only produce sound results for 
analysis tasks being reducible to plain (location-)reachability problems thus
being too imprecise for checking timed bisimilarity~\cite{Weise1997}.

\paragraph{Conceptual Contributions.}
In this paper, we propose a novel characterization 
of TA semantics, called \emph{bounded zone-history graphs}.
\emph{Zone histories} enrich plain zone graphs exactly by the additional
information required for sound timed bisimilarity-checking, yet still yielding
a finite representation of TA semantics.
However, in order to control the size of
bounded zone-history graphs in case of larger input models, our 
approach further incorporates a \emph{bound parameter} $b$ to restrict
the length of histories. 
This bound parameter thus facilitates an adjustable trade-off 
between precision and scalability of timed-bisimilarity checking.
Our technique further handles \emph{non-deterministic} TA and 
supports \emph{weak} and \emph{strong} bisimilarity 
of non-deterministic timed (safety) automata with silent $\tau$-moves.

\paragraph{Tool Support and Reproducibility.}
Our tool \timbrcheck{} supports the \uppaal{}
file format for input models 
and is available on our complementary web page\footnote{\url{https://www.es.tu-darmstadt.de/timbrcheck/}}.
This web page also contains all experimental data and further information
for reproducing the evaluation results.
Additionally, we provide a rich collection of test cases 
(\ie{} pairs of input models) constituting
particularly sophisticated 
TA fragments which we used to exhaustively test our tool implementation.

\paragraph{Experimental Evaluation.}
Our experimental results gained from applying \timbrcheck{} to
a collection of community benchmarks~\cite{Alur1993,Lindahl2001,Jensen1996,Collomb2001,Havelund1997} provide insights into trade-offs 
between precision and efficiency of checking timed bisimilarity using bounded zone-history graphs.
In particular, our results indicate, that a value of 3 for bound parameter $b$ appears to be a reasonable trade-off 
between precision and scalability for the subject systems under consideration.
Moreover, as expected, checking TA with non-deterministic behavior requires considerably more computational 
effort than deterministic cases.

\paragraph{Related Work.}
The notion of timed bisimulation goes back to the works of
Moller and Tofts~\cite{Moller1990} as well as Yi~\cite{Yi1990}
both originally defined on real-time extensions of the process algebra CCS.
Similarly, Nicollin and Sifakis~\cite{Nicollin1994} define 
timed bisimulation on ATP (Algebra of Timed Processes).
However, none of these works initially incorporated a technique for
effectively checking bisimilarity.
The pioneering work of \v{C}er{\=a}ns~\cite{Cerans1992}
includes the first decidability proof of timed bisimulation on 
TA, by providing a finite characterization of bisimilarity-checking 
on a finite representation of TA semantics, called region graphs.
The improved (\ie{} less space-consuming) 
approach of Weise and Lenzkes~\cite{Weise1997} 
employs a variation of zone graphs, called FBS graphs, which also builds the basis
for our notion of zone-history graphs.
Guha \etal{}~\cite{Guha2012prebisimulation,Guha2013}
also follow a zone-based approach for bisimilarity-checking on TA as well as
the weaker notion of timed prebisimilarity, by employing so-called zone-valuation 
graphs and the notion of spans as also used in our approach.
Moreover, Tanimoto \etal{}~\cite{Tanimoto2004} employ 
timed bisimulation to check if a given behavioral 
abstraction preserves time-critical system behavior.

Nevertheless, all these approaches neither facilitate an adjustable trade-off between precision and scalability for checking timed bisimilarity nor
provide any practical tool support.
The only \emph{currently available} tool for checking timed bisimilarity we are aware of 
is called \textsc{Caal}~\cite{Andersen2015caal} which is, however, inherently incomplete
as it does not utilize a finite representation of TA semantics.

%% file: sections/background.tex
%
\section{Preliminaries}\label{section:background}

In this section, we introduce the notational foundations of timed
automata and timed bisimulation.

\subsection{Timed Automata}\label{subsection:ta}

\paragraph{Syntax.}
A \emph{timed automaton (TA)} consists of finite state-transition
graph whose states are called \emph{locations} 
(including a distinguished \emph{initial location}) and whose 
edges, denoting transitions between locations, are called \emph{switches}~\cite{Alur1990}.
Switches are either labeled with names from a finite alphabet $\Sigma$ of \emph{visible actions}, or
by a distinguished symbol $\tau\not\in\Sigma$, denoting \emph{internal actions} (silent moves).
We range over $\Sigma$ by $\sigma$ and over $\Sigma_{\tau}=\Sigma\cup\{\tau\}$ by $\mu$.

A TA further consists of a finite set $C$ of \emph{clocks}, defined
over a numerical \emph{clock domain} $\mathbb{T}_{\emph{C}}$ (\eg{} $\mathbb{T}_{\emph{C}}=\mathbb{N}_{0}$
for modeling \emph{discrete time} and $\mathbb{T}_{\emph{C}}=\mathbb{R}_{+}$ for modeling \emph{dense time}),
where we consider $\mathbb{T}_C=\mathbb{N}_0$ in all upcoming examples.
Clocks may be considered as constantly and synchronously increasing
yet independently resettable variables over $\mathbb{T}_{\emph{C}}$.
Clocks allow for measuring and restricting time intervals
corresponding to durations---or delays between occurrences---of actions in \emph{valid runs} of a TA.
Those restrictions are expressed by \emph{clock constraints} $\varphi$
to denote \emph{guards} for switches and \emph{invariants} for locations.
Guards restrict time intervals in which particular switches are enabled, 
whereas invariants restrict time intervals in which TA runs
are permitted to reside in particular locations.
In addition, each switch is labeled with a subset of clocks $R\subseteq C$ 
to be \emph{reset}.

\begin{definition}[Timed Automaton]\label{def:ta}
  A \emph{TA} is a tuple $\left(L,\ell_0,\Sigma,C,I,E\right)$, where
  \begin{itemize}
    \item $L$ is a finite set of \emph{locations} with \emph{initial location} $\ell_0\in L$,
    \item $\Sigma$ is a finite set of \emph{actions} such that $\tau\not\in\Sigma$,
    \item $C$ is a finite set of \emph{clocks} such that $C\cap\Sigma_\tau=\emptyset$,
    \item $I:L\rightarrow\mathcal{B}(C)$ is a function assigning \emph{invariants} to locations, and
    \item $E\subseteq L\times\mathcal{B}(C)\times\Sigma_{\tau}\times 2^{C}\times L$ 
    is a relation defining \emph{switches}.
  \end{itemize}
  The set $\mathcal{B}(C)$ of \emph{clock constraints} $\varphi$ over $C$ 
  is inductively defined as
  $$\varphi:=\mathsf{true}\mid c\sim n\mid c-c'\sim n\mid\varphi\land\varphi,~\emph{where}~{\sim}\in\{<,\leq,\geq,>\}, c, c'\in C, n\in\mathbb{T}_\textit{C}.$$
\end{definition}

We denote TA defined over sets $C$ and $\Sigma$ 
by $\mathcal{A}$ where we may omit an explicit mentioning of $C$ and/or $\Sigma$ if clear
from the context.
We further denote switches $(\ell,g,\mu,R,\ell')\in E$ by
$\ell\xrightarrow{g,\mu,R}\ell'$ for convenience.
Clock constraints neither contain operators for equality nor 
disjunction as both are equivalently expressible by the given grammar
(\eg{} switch guard $x=2$ may be expressed by $x\leq2\land x\geq2$, and $x<2\lor x>2$ 
may be expressed by duplicating the switch, one labeled with guard $x<2$ and one with $x>2$, respectively).

Moreover, we consider \emph{diagonal-free} TA
with clock constraints only containing atomic constraints of the form $c\sim n$
as for every TA, a language-equivalent diagonal-free TA can be constructed~\cite{Berard1998}.
Hence, we include \emph{difference constraints} $c-c'\sim n$ into $\mathcal{B}(C)$
solely for the sake of a concise representation of our subsequent constructions.
Similarly, we assume location invariants being unequal to $\mathsf{true}$ 
to be \emph{downward-closed} (\ie{} only having clauses of the form $c \leq n$ or $c < n$).
However, as two actual restrictions, we limit our considerations 
to (1)~constants $n\in \mathbb{Q}_{0}$ as real-valued bounds
would obstruct fundamental decidability properties of TA, as well as to
(2)~so-called timed \emph{safety} automata not including distinguished 
\emph{acceptance locations} for employing Büchi accepting-trace semantics for 
infinite TA runs~\cite{Henzinger1994,Alur1990}.

\paragraph{Semantics.}
The operational semantics of a given TA, defining all its valid (timed) runs,
may be defined in terms of \emph{Timed Labeled Transition Systems (TLTS)}~\cite{Henzinger1991}.
A TLTS \emph{state} is a pair $\langle \ell, u\rangle$ 
of active location $\ell \in L$ and \emph{clock valuation} $u\in C\rightarrow \mathbb{T}_{\emph{C}}$
assigning to each clock $c\in C$ the amount of time $u(c)$ elapsed since the last reset of $c$.
Thereupon, TLTS comprise two kinds of \emph{transitions}:
(1)~passage of time of duration $d \in \mathbb{T}_{C}$ while (inactively) residing in location $\ell$, 
leading to an updated clock valuation $u'$, and
(2)~instantaneous executions of switches $\ell\xrightarrow{g,\mu,R}\ell'$, 
leading from location $\ell$ to $\ell'$, accompanied by an occurrence of action $\mu\in\Sigma_{\tau}$.

Given clock valuation $u$, by $u+d$ with $d\in \mathbb{T}_{C}$,
we denote the \emph{updated clock valuation} mapping 
each clock $c\in C$ to the new value $u(c)+d$.
By $[R\mapsto 0]u$, with $R\subseteq C$, we further denote the updated
clock valuation mapping each clock $c\in R$ to value 0 (\emph{clock reset})
while preserving the values $u(c')$ of all other clocks $c'\in C\setminus R$.
Finally, by $u\in \varphi$, we denote that clock valuation $u$ 
\emph{satisfies} clock constraint $\varphi\in \mathcal{B}(C)$.
Concerning $\tau$-labeled transitions, we distinguish
between \emph{strong} and \emph{weak} TLTS semantics, where $\tau$-transitions are invisible in the latter case.

\begin{definition}[Timed Labeled Transition System]\label{def:tlts-semantics}
  The TLTS of TA $\mathcal{A}$ over $\Sigma$ is a tuple $(S,s_0,\hat{\Sigma},\twoheadrightarrow)$, where
  \begin{itemize}
    \item $S=L\times(\mathcal{C}\rightarrow\mathbb{T}_{\emph{C}})$ is a set of \emph{states} with \emph{initial state} $s_0=\langle \ell_{0}, [C\mapsto 0]\rangle\in S$,
    \item $\hat\Sigma=\Sigma\cup\Delta$ is a set of \emph{transition labels}, where
        $\Delta=\mathbb{T}_{\emph{C}}$ with $(\Sigma \cup \{\tau\})\cap\Delta=\emptyset$, and
    \item ${\twoheadrightarrow}\subseteq S\times(\hat{\Sigma}\cup\{\tau\})\times S$ 
  is a set of \emph{strong transitions} being the least relation satisfying the rules:
  \begin{itemize}
    \item $\angles{\ell,u}\xtwoheadrightarrow{d}\angles{\ell,u+d}$ if $(u+d)\in I(\ell)$ for $d\in\mathbb{T}_{\emph{C}}$, and
    \item $\angles{\ell,u}\xtwoheadrightarrow{\mu}\angles{\ell',u'}$ if $\ell\xrightarrow{g,\mu,R}\ell'$, $u\in g$, $u'=[R\mapsto 0]u$, $u'\in I(\ell')$ 
    and $\mu\in(\Sigma\cup\{\tau\})$.
  \end{itemize}
  \end{itemize}
  By ${\Twoheadrightarrow}\subseteq S\times \hat{\Sigma}\times S$, we denote
  a set of \emph{weak transitions} being the least relation satisfying the rules:
  \begin{itemize}
  \item $s\xTwoheadrightarrow{\sigma}s'$ if $s\xtwoheadrightarrow{\tau^{n}}s_1\xtwoheadrightarrow{\sigma}s_2\xtwoheadrightarrow{\tau^{m}}s'$ with $n, m\in \mathbb{N}_{0}$,
  \item $s\xTwoheadrightarrow{d}s'$ if $s\xtwoheadrightarrow{d}s'$,
  \item $s\xTwoheadrightarrow{0}s'$ if $s\xtwoheadrightarrow{\tau^{n}}s'$ with $n\in \mathbb{N}_{0}$, and
  \item $s\xTwoheadrightarrow{d+d'}s'$ if $s\xTwoheadrightarrow{d}s''$ and $s''\xTwoheadrightarrow{d'}s'$.
\end{itemize}
\end{definition}

We only consider TA with \emph{strongly convergent} TLTS (\ie{} without infinite $\tau$-sequences)
and refer to the TLTS semantics of TA $\mathcal{A}$ as $\mathcal{S}_{\mathcal{A}}$
or simply as $\mathcal{S}$ if clear from the context.
In addition, if not explicitly stated, we consider strong TLTS semantics, where
the corresponding weak version can by obtained by replacing $\twoheadrightarrow$ by $\Twoheadrightarrow$
in the following.

\begin{example}\label{example:ta}
Figure~\ref{fig:ta-example} shows two sample TA specifying (simplified) coffee machines 
with corresponding TLTS extracts shown in 
Figures~\ref{fig:tlts-example-machine} and~\ref{fig:tlts-example-machine-prime}.
\begin{figure}[tp]
  \hfill
  \subfloat[Coffee Machine]{\label{fig:ta-example-machine}\input{figures/ta-example-machine.tex}}
  \hfill
  \subfloat[Coffee Machine$'$]{\label{fig:ta-example-machine-prime}\input{figures/ta-example-machine-prime.tex}}
  \hfill\strut

  \hfill
  \subfloat[TLTS of Fig.~\ref{fig:ta-example-machine}]{\label{fig:tlts-example-machine}\input{figures/tlts-example-machine.tex}}
  \hfill
  \subfloat[TLTS of Fig.~\ref{fig:ta-example-machine-prime}]{\label{fig:tlts-example-machine-prime}\input{figures/tlts-example-machine-prime.tex}}
  \hfill\strut
  \caption{TA of Two Similar Coffee Machines (Figs.~\ref{fig:ta-example-machine} and~\ref{fig:ta-example-machine-prime}) and TLTS (Figs.~\ref{fig:tlts-example-machine} and~\ref{fig:tlts-example-machine-prime})}\label{fig:ta-example}
\end{figure}
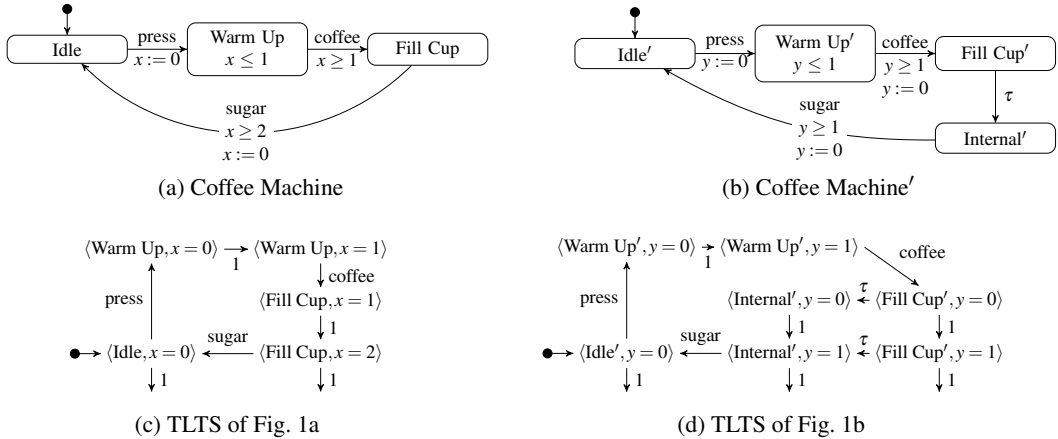
In state $\angles{\text{Warm Up},x=0}$, we can only let further time pass 
whereas in $\angles{\text{Warm Up},x=1}$, we have to 
choose \emph{coffee} due to the invariant.
In contrast, as neither location \emph{Idle} nor \emph{Fill Cup} 
has an invariant, we may wait for an unlimited amount of time 
thus resulting in infinitely many consecutive TLTS states.
Further note that the TLTS in Fig.~\ref{fig:tlts-example-machine-prime} contains a $\tau$-transition
which is only visible in the strong case.
\end{example}

\subsection{Timed Bisimulation}

We next revisit the notion of \emph{timed bisimulation} 
to semantically compare different TA defined over the same alphabet.
A timed (bi-)simulation relation may be defined
by directly adapting the classical notion of (bi-)simulation on LTS to TLTS.
State $s'$ of TLTS $\mathcal{S}_{\mathcal{A}'}$ \emph{timed simulates} state $s$
of TLTS $\mathcal{S}_{\mathcal{A}}$ if every transition enabled in $s$, either labeled 
with action $\mu \in \Sigma_{\tau}$ or delay $d\in \Delta$, is also enabled in 
$s'$ and the target state in $\mathcal{S}_{\mathcal{A}'}$, again, timed
simulates the respective target state in $\mathcal{S}_{\mathcal{A}}$.
Hence, TA $\mathcal{A}'$ \emph{timed simulates} $\mathcal{A}$ if initial state
$s_0'$ \emph{timed simulates} initial state $s_0$ and $\mathcal{A}'$ and $\mathcal{A}$ are \emph{timed bisimilar}
if the timed simulation relation is symmetric.

\begin{definition}[Timed Bisimulation~\cite{Weise1997}]\label{def:ta-bisim}
Let $\mathcal{A}$, $\mathcal{A}'$ be TA over $\Sigma$
with $C\,\cap\,C'=\emptyset$ and $\mathcal{R}\subseteq S \times S'$ such that for all $(s_1,s_1')\in \mathcal{R}$
it holds that
\begin{itemize}
  \item if $s_1 \xtwoheadrightarrow{\mu}s_2$ with $\mu\in \Sigma_{\tau}$,
then $s_1' \xtwoheadrightarrow{\mu}s_2'$ and $(s_2,s_2')\in \mathcal{R}$ and
  \item if $s_1 \xtwoheadrightarrow{d}s_2$ with $d\in \Delta$
then $s_1' \xtwoheadrightarrow{d}s_2'$ with $(s_2,s_2')\in \mathcal{R}$.
\end{itemize}
$\mathcal{A}'$ \emph{(strongly) timed simulates} $\mathcal{A}$, denoted $\mathcal{A}\sqsubseteq\mathcal{A}'$, iff $(s_0,s_0')\in \mathcal{R}$.
In addition, $\mathcal{A}'$ and $\mathcal{A}$ are \emph{(strongly) timed bisimilar}, denoted $\mathcal{A}\simeq\mathcal{A}'$, iff $\mathcal{R}$ is symmetric.
\end{definition}

\emph{Weak} timed (bi-)simulation can, again, be 
obtained by replacing $\twoheadrightarrow$ with $\Twoheadrightarrow$ in all
definitions (which we will omit if not relevant).

\begin{example}
Consider, again, $\mathcal{A}$ and $\mathcal{A}'$ in Figs.~\ref{fig:ta-example-machine} 
and~\ref{fig:ta-example-machine-prime}.
Strong timed (bi-)simulation does not hold between both models 
due to the $\tau$-step in $\mathcal{A}'$.
In contrast, for the weak case, we have $\mathcal{A}\sqsubseteq\mathcal{A}'$ as every action and delay
of $\mathcal{A}$ is also permitted by $\mathcal{A}'$ (\cf{} TLTS in Figs.~\ref{fig:tlts-example-machine} 
and~\ref{fig:tlts-example-machine-prime}).
Similarly, $\mathcal{A}'\sqsubseteq\mathcal{A}$ also holds
such that $\mathcal{A}$ and $\mathcal{A}'$ are weakly timed bisimilar.
\end{example}

We conclude this section by repeating 
the well-known result that systems being strong (timed) similar are also weak (timed) similar.

\begin{lemma}\label{lemma:sim-inclusions}
If $\mathcal{A}'$ strongly timed simulates $\mathcal{A}$, then $\mathcal{A}'$ weakly timed simulates $\mathcal{A}$~\cite{Yi1990}.
\end{lemma}

\begin{proof}
We prove Lemma~\ref{lemma:sim-inclusions} by contradiction.
Assume TA $\mathcal{A}$ and $\mathcal{A}'$ with $\mathcal{A}'$ strongly timed simulating $\mathcal{A}$ and $\mathcal{A}'$ \emph{not} weakly timed simulating $\mathcal{A}$.
In this case, we require TLTS states $\angles{\ell_1,u_1}\in S$ and $\angles{\ell_1',u_1'}\in S'$ being reachable by a $\tau$-step such that for each $\angles{\ell_1',u_1'}\xtwoheadrightarrow{\eta}\angles{\ell_2',u_2'}\in{\twoheadrightarrow}'$ with $\eta\in\hat\Sigma$ there exists a $\angles{\ell_1,u_1}\xtwoheadrightarrow{\eta}\angles{\ell_2,u_2}\in{\twoheadrightarrow}$.
Due to the definition of weak transitions (see Def.~\ref{def:tlts-semantics}), we also require a transition $\angles{\ell_1,u_1}\xtwoheadrightarrow{\eta}\angles{\ell_2,u_2}\in{\twoheadrightarrow}$ not being enabled in $\angles{\ell_1,u_1}$ to prove that $\mathcal{A}'$ strongly timed simulates $\mathcal{A}$ and $\mathcal{A}'$ weakly timed simulates $\mathcal{A}$.
However, as these two assumptions are contradicting, it holds that $\mathcal{A}'$ weakly timed simulates $\mathcal{A}$ if $\mathcal{A}'$ strongly timed simulates~$\mathcal{A}$.
\end{proof}

%% file: figures/ta-example-machine.tex

\scalebox{\scalefactor}{
\begin{tikzpicture}

\node[labeledroundedstateTA, initial, initial where=above] (idle) {Idle};
\node[labeledroundedstateTA, right=of idle, align=center] (warmup) {Warm Up\\$x\leq 1$};
\node[labeledroundedstateTA, right=of warmup] (fillcup) {Fill Cup};

\draw[transition] (idle) to node[pos=.5, align=center]{press\\$x:=0$} (warmup);
\draw[transition] (warmup) to node[pos=.5, align=center]{coffee\\$x\geq1$} (fillcup);
\draw[transition, bend left=45] (fillcup) to node[pos=.5, descr, align=center]{sugar\\$x\geq 2$\\$x:=0$} (idle);

\end{tikzpicture}
}

%% file: figures/ta-example-machine-prime.tex

\scalebox{\scalefactor}{
\begin{tikzpicture}

\node[labeledroundedstateTA, initial, initial where=above] (idle) {Idle$'$};
\node[labeledroundedstateTA, right=of idle, align=center] (warmup) {Warm Up$'$\\$y\leq 1$};
\node[labeledroundedstateTA, right=of warmup] (fillcup) {Fill Cup$'$};
\node[labeledroundedstateTA, below=of fillcup, yshift=.33em] (internal) {Internal$'$};

\draw[transition] (idle) to node[pos=.5, align=center]{press\\$y:=0$} (warmup);
\draw[transition] (warmup) to node[pos=.5, align=center]{\strut\\coffee\\$y\geq1$\\$y:=0$} (fillcup);
\draw[transition] (fillcup) to node[auto]{$\tau$} (internal);
\draw[transition, bend left=15] (internal) to node[pos=.4, descr, align=center]{sugar\\$y\geq 1$\\$y:=0$} (idle);

\end{tikzpicture}
}

%% file: figures/tlts-example-machine.tex

\scalebox{\scalefactor}{
\begin{tikzpicture}[node distance=.4]

\node[labeledstate] (warmupzero) {$\angles{\text{Warm Up},x=0}$};
\node[labeledstateempty, below=of warmupzero] (idleempty) {$\angles{\text{Fill Cup},x=1}$};
\node[labeledstate, below=of idleempty, initial, initial where=left] (idlezero) {$\angles{\text{Idle},x=0}$};
\node[labeledstate, below=of idlezero] (idleinfinite) {};
\node[labeledstate, right=of warmupzero] (warmupone) {$\angles{\text{Warm Up},x=1}$};
\node[labeledstate, below=of warmupone] (fillcupone) {$\angles{\text{Fill Cup},x=1}$};
\node[labeledstate, below=of fillcupone] (fillcuptwo) {$\angles{\text{Fill Cup},x=2}$};
\node[labeledstate, below=of fillcuptwo] (fillcupinfinite) {};

\draw[transition] (idlezero) to node[auto]{press} (warmupzero);
\draw[transition] (idlezero) to node[auto]{1} (idleinfinite);
\draw[transition] (warmupzero) to node[auto, swap]{1} (warmupone);
\draw[transition] (warmupone) to node[auto]{coffee} (fillcupone);
\draw[transition] (fillcupone) to node[auto]{1} (fillcuptwo);
\draw[transition] (fillcuptwo) to node[auto, swap]{sugar} (idlezero);
\draw[transition] (fillcuptwo) to node[auto]{1} (fillcupinfinite);

\end{tikzpicture}
}

%% file: figures/tlts-example-machine-prime.tex

\scalebox{\scalefactor}{
\begin{tikzpicture}[node distance=.4]

\node[labeledstate] (warmupzero) {$\angles{\text{Warm Up$'$},y=0}$};
\node[labeledstateempty, below=of warmupzero] (idleempty) {$\angles{\text{Internal$'$},y=1}$};
\node[labeledstate, below=of idleempty, initial, initial where=left] (idlezero) {$\angles{\text{Idle$'$},y=0}$};
\node[labeledstate, below=of idlezero] (idleinfinite) {};
\node[labeledstate, right=of warmupzero, xshift=-.5em] (warmupone) {$\angles{\text{Warm Up$'$},y=1}$};
\node[labeledstate, below=of warmupone] (internalzero) {$\angles{\text{Internal$'$},y=0}$};
\node[labeledstate, below=of internalzero] (internalone) {$\angles{\text{Internal$'$},y=1}$};
\node[labeledstate, below=of internalone] (internalinfinite) {};

\node[labeledstate, right=of internalzero, xshift=-.5em] (fillcupzero) {$\angles{\text{Fill Cup$'$},y=0}$};
\node[labeledstate, below=of fillcupzero] (fillcupone) {$\angles{\text{Fill Cup$'$},y=1}$};
\node[labeledstate, below=of fillcupone] (fillcupinfinite) {};

\draw[transition] (idlezero) to node[auto]{press} (warmupzero);
\draw[transition] (idlezero) to node[auto]{1} (idleinfinite);
\draw[transition] (warmupzero) to node[auto, swap]{1} (warmupone);
\draw[transition] (warmupone.east) to node[auto]{coffee} (fillcupzero);
\draw[transition] (internalzero) to node[auto]{1} (internalone);
\draw[transition] (internalone) to node[auto, swap]{sugar} (idlezero);
\draw[transition] (internalone) to node[auto]{1} (internalinfinite);
\draw[transition] (fillcupzero) to node[auto, swap]{$\tau$} (internalzero);
\draw[transition] (fillcupone) to node[auto, swap]{$\tau$} (internalone);
\draw[transition] (fillcupzero) to node[auto]{1} (fillcupone);
\draw[transition] (fillcupone) to node[auto]{1} (fillcupinfinite);

\end{tikzpicture}
}

%% file: sections/timed-bisimulation.tex

\section{Checking Timed Bisimulation with Bounded Zone-History Graphs}\label{section:timed-bisimulation}

As TLTS are, in general, infinite-state and infinitely-branching LTS, 
they are only of theoretical interest, but 
do not facilitate effective timed (bi-)similarity checking.
In~\cite{Cerans1992}, a finite, yet often unnecessarily space-consuming 
characterization of timed bisimilarity 
is given using \emph{region graphs} instead of TLTS.
In contrast, Weise and Lenzkes~\cite{Weise1997} use
so-called \emph{full backward stable (FBS) graphs}, 
an adaption of the symbolic 
\emph{zone-graph} representation~\cite{Dill1989} of TA semantics enriched by transition labels.
Zone graphs are, in most cases, less 
space-consuming than region graphs.
We will also built upon FBS graphs in the following, but propose a novel 
definition, called \emph{(bounded) zone-history graphs}, to permit
a more concise characterization and 
scalable checking of timed \mbox{(bi-)}simulation.

\subsection{Zone Graphs}

A \emph{symbolic state} of TA $\mathcal{A}$ is a pair $\angles{\ell,\varphi}$ 
consisting of a location $\ell\in L$ and a \emph{zone} $\varphi\in\mathcal{B}(C)$, where $\varphi$
represents the maximum set $D = \{u:C\rightarrow \mathbb{T}_{C}\,|\,u \in \varphi\}$ of clock valuations $u$ 
satisfying clock constraint $\varphi$.
Hence, symbolic state $\angles{\ell,\varphi}$ comprises all TLTS 
states $\angles{\ell,u}\in S$ of $\mathcal{S}_\mathcal{A}$ 
with $u\in D$, where we may use $\varphi$ and $D$ interchangeably in the following.
The construction of a zone graph for a timed automaton is 
based on two operations on zones:

\begin{itemize}
  \item $D^\uparrow=\{u+d\mid u\in D,d\in\mathbb{T}_C\}$ denotes the \emph{future} of zone $D$, and
  \item $R(D)=\{[R\mapsto0]u\mid u\in D\}$ denotes the application of a set of \emph{clock resets} $R\subseteq C$ on zone $D$.
\end{itemize}

By $D_0$, we denote the \emph{initial zone} in which all clock values are mapped to constant $0$.
For each switch $\ell\xrightarrow{g,\mu,R}\ell'$, a corresponding
transition $\angles{\ell,D}\xrightsquigarrow{\mu}\angles{\ell',D'}$
is added with target zone $D'$ derived from source zone $D$
by considering the future $D^\uparrow$ of $D$, further restricted
by rgw switch guard $g$, the location invariants of $\ell$ and $\ell'$ as well as the clock resets $R$.

\begin{definition}[Zone Graph]\label{def:zone-graph}
The \emph{zone graph} of TA $\mathcal{A}$ over $\Sigma$ 
is a tuple $(\mathcal{Z},z_0,\Sigma,\rightsquigarrow)$, where
\begin{itemize}
  \item $\mathcal{Z}=L\times\mathcal{B}(C)$ is a set of \emph{symbolic states}
  with \emph{initial state} $z_0=\angles{\ell_0,D_0}$,
  \item $\Sigma$ is a set of \emph{actions}, and
  \item ${\rightsquigarrow}\subseteq\mathcal{Z}\times\Sigma_\tau\times\mathcal{Z}$ 
  is the least relation satisfying the rule:\\
  $\angles{\ell,D}\xrightsquigarrow{\mu}\angles{\ell',D'}$ if $\ell\xrightarrow{g,\mu,R}\ell'$ and $D'=R(D^\uparrow\land g\land I(\ell))\land I(\ell')$.
\end{itemize}
\end{definition}

Although zone graphs according to Def.~\ref{def:zone-graph} are, again, not necessarily finite, 
an equivalent, finite zone-graph representation for any given TA can be obtained
(1) by constructing an equivalent \emph{diagonal-free} TA only
containing atomic clock constraints of the form $x\sim r$~\cite{Berard1998}, 
and (2) by constructing for this TA a \emph{$k$-bounded} 
zone-graph representation according to Def.~\ref{def:zone-graph} where all zones being bound by a maximum global \emph{clock ceiling} $k$ 
using $k$-normalization~\cite{Rokicki1994,Pettersson1999}.

%

\subsection{Spans}

The comparison of zones of two different TA during timed bisimilarity-checking is based on the notion of \emph{spans}~\cite{Guha2012deciding}.
The span of clock $c\in C$ in zone $D$ is the interval $(\textit{lo},\textit{up})$ between
the minimum valuation $\textit{lo}$ and maximum valuation $\textit{up}$ of $c$ in $D$.
The span of zone $D$ is the least interval covering
the spans of all clocks in $D$.
By $\infty$, we denote upward-open intervals (\ie{} $d<\infty$ for all $d\in \mathbb{T}_{C}$), 
where $\infty$ behaves in calculations as usual.

We further introduce two operators for comparing spans $\textit{sp}_1$ and $\textit{sp}_2$: $\textit{sp}_1\preceq\textit{sp}_2$ 
denotes that $\textit{sp}_1$ is \emph{contained} in $\textit{sp}_2$, whereas $\textit{sp}_1\leq\textit{sp}_2$ denotes that the relative \emph{length} 
of $\textit{sp}_1$ is shorter than $\textit{sp}_2$.
Please note that we overload the notion of spans to likewise refer the set of elements within the interval defined by a span.
Hence, span $(\textit{lo},\textit{up})$ denotes the set of elements $n$ with $n\geq\textit{lo}\land n\leq\textit{up}$.

\begin{definition}[Span]\label{def:span}
Given zone $D$ and $c\in C$, we use the following notations.
\begin{itemize}
    \item $\emph{span}(c,D)=(\textit{lo},\textit{up})\in\mathbb{T}_C\times(\mathbb{T}_C\cup\{\infty\})$ is the smallest 
    interval such that $\forall u\in D: u(c)\geq\textit{lo}\land u(c)\leq\textit{up}$.
    \item $\emph{span}(\textit{lo},\textit{up})=\{n\in\mathbb{T}_C\mid n\geq\textit{lo}\land n\leq\textit{up}\}$.
  \item $(\textit{lo},\textit{up})\prec(\textit{lo}',\textit{up}') \Leftrightarrow\textit{lo}>\textit{lo}'\land\textit{up}<\textit{up}'$.
  \item $(\textit{lo},\textit{up})\preceq(\textit{lo}',\textit{up}') \Leftrightarrow\textit{lo}\geq\textit{lo}'\land\textit{up}\leq\textit{up}'$.
  \item $(\textit{lo},\textit{up})\leq(\textit{lo}',\textit{up}')\Leftrightarrow\textit{up}-\textit{lo}\leq\textit{up}'-\textit{lo}'$.
\end{itemize}
\end{definition}

Based on the notion of spans, we are able compare timing constraints of action occurrences 
of two different TA independent of the names of locations and clocks.
However, due to non-observability of clock resets,
it is not sufficient for timed (bi-)simulation checking to just compare
spans of pairs of potentially similar symbolic states one-by-one
as will be illustrated by the following example.

\begin{example}
Considering TA $\mathcal{A}$ and $\mathcal{A}'$ in 
Figs.~\ref{fig:reset-history-a} and~\ref{fig:reset-history-a-prime},
the span of action $a$ is $(0,2)$ in both TA due to the switch guards.
Additionally, the span for action $b$ is $(0,5)$ in both TA.
However, in $\mathcal{A}$, we may only wait for 5 time units before performing $b$ if we 
have instantaneously (\ie{} with 0 delay) performed $a$ before, whereas in $\mathcal{A}'$, the delay for 
performing $b$ is independent of previous delays due to the reset of $z$.
Hence, $\mathcal{A}\simeq\mathcal{A}'$ does not hold.
\begin{figure}[tp]
  \hfill
  \subfloat[$\mathcal{A}$]{\label{fig:reset-history-a}\input{figures/reset-history-a.tex}}
  \hfill
  \subfloat[$\mathcal{A}'$]{\label{fig:reset-history-a-prime}\input{figures/reset-history-a-prime.tex}}
  \hfill
  \subfloat[$\mathcal{ZH}$]{\label{fig:reset-history-graph-a}\input{figures/reset-history-graph-a.tex}}
  \hfill
  \subfloat[$\mathcal{ZH}'$]{\label{fig:reset-history-graph-a-prime}\input{figures/reset-history-graph-a-prime.tex}}
  \hfill\strut
  \caption{False Positive using Plain Zone Graphs for Checking Timed Bisimilarity}\label{fig:reset-history}
\end{figure}
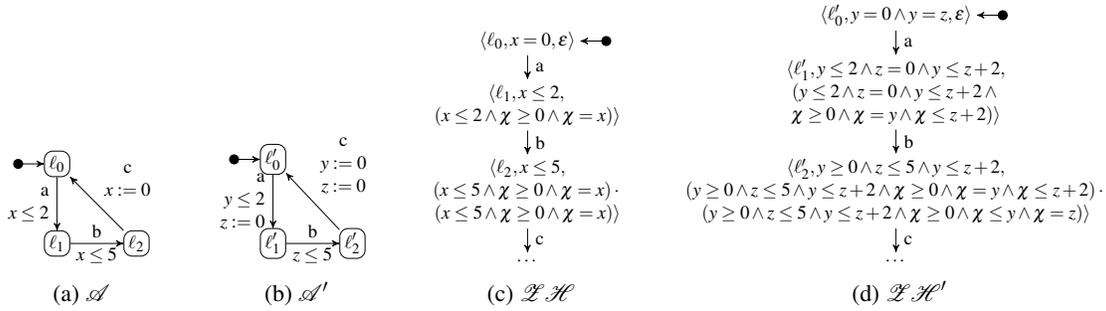
\end{example}

In~\cite{Weise1997}, this issue is tackled by further considering so-called 
\emph{good sequences} of FBS graphs in a separate post-check.
In contrast, we propose an alternative solution being more aligned
with the concepts of (bi-)simulation equivalence relations 
on state-transition graphs (\ie{} by enriching symbolic states
with additionally discriminating information).

\subsection{Zone-History Graphs}

Similar to the notion of \emph{causal history} as, for instance, proposed
for history-preserving event-structure semantics~\cite{Baldan1999}, we extend symbolic states $\angles{\ell,D}$ 
to triples $\angles{\ell,D,\mathcal{H}}$ further comprising a \emph{zone history} $\mathcal{H}\in \mathcal{B}(C)^*$
to memorize \emph{sequences of clock constraints} corresponding to the zones of predecessor states.
When stepping from zone $D$ to zone $D'$, the history
$\mathcal{H}$ is \emph{updated} to $\mathcal{H}'$ according to the updates applied to $D$ leading to $D'$.
By introducing a fresh clock $\chi\not\in C$ which is never explicitly reset, we measure the respective \emph{spans} of histories $\mathcal{H}$ in order to compare the sequences of intervals through which the current states are reachable from their predecessors.
By $H\cdot\mathcal{H}'$ and $\mathcal{H}'\cdot H$, respectively, we denote the \emph{concatenation} of further elements $H$ in front of, or after, history sequences $\mathcal{H}'$, where $\epsilon$ denotes the \emph{empty} sequence with $\mathcal{H}\cdot \epsilon = \epsilon \cdot \mathcal{H} = \mathcal{H}$.

\begin{definition}[Zone History]\label{def:zone-history}
Let $\mathcal{H}\in\mathcal{B}(C\cup\{\chi\})^*$ with $\chi\notin C$ be a \emph{zone history}.
The \emph{update} of history $\mathcal{H}$ for a switch $\ell\xrightarrow{g,\mu,R}\ell'$ 
leading from zone $D$ to $D'=R(D^\uparrow\land g\land I(\ell))\land I(\ell')$ is recursively defined as
\begin{itemize}
  \item $\emph{update}(\mathcal{H},D,D')=R(H^\uparrow\land g\land I(\ell))\land I(\ell')\cdot\emph{update}(\mathcal{H}',D,D')$ if $\mathcal{H}=H\cdot\mathcal{H}'$,
  \item $\emph{update}(\mathcal{H},D,D')=R((D\land\chi=0)^\uparrow \land g\land I(\ell)) \land I(\ell')$ if $\mathcal{H}=\epsilon$.
\end{itemize}
\end{definition}

We are now ready to define a \emph{zone-history graph} of TA $\mathcal{A}$ 
by extending plain \emph{zone graphs} (see Def.~\ref{def:zone-graph}) with zone histories.
The initial state $z_0=\angles{\ell_0,D_0,\epsilon}$ comprises initial location
$\ell_0$, initial zone $D_0$ and the \emph{empty} history.
The target state $\angles{\ell',D',\mathcal{H}'}$ of a transition
$\angles{\ell,D,\mathcal{H}}\xrightsquigarrow{\mu}\angles{\ell',D',\mathcal{H}'}$
corresponding to a switch $\ell\xrightarrow{g,\mu,R}\ell'$ is reached by updating zone $D$ to 
$D'$ as described before, and by additionally updating 
history $\mathcal{H}$ to $\mathcal{H}'=\emph{update}(\mathcal{H},D,D')$.
We write $z_i$ to refer to 
states $\angles{\ell_i,D_i,\mathcal{H}_i}$ (\ie{} every element of state $z_i$ has index $i$).

Please note that this construction only serves as a theoretical baseline
as it would, again, yield an infinite zone-history graph whenever the respective TA contains cyclic paths (thus 
leading to an infinitely growing history-component of zones).
In order to handle cyclic behavior, we will present an algorithm for pruning (possibly infinite) 
zone-history graphs into finite ones for effectively checking timed bisimilarity.

\begin{definition}[Zone-History Graph]\label{def:zone-history-graph}
The \emph{zone-history graph} of a TA $\mathcal{A}$ with $\chi\notin C$ over $\Sigma$ 
is a tuple $(\mathcal{Z},z_0,\Sigma,\rightsquigarrow)$, where
\begin{itemize}
  \item $\mathcal{Z}=L\times\mathcal{B}(C)\times\mathcal{B}(C\cup\{\chi\})^*$ is a set of \emph{symbolic states}
  with $z_0=\angles{\ell_0,D_0,\epsilon}$,
  \item $\Sigma$ is a set of \emph{actions}, and
  \item ${\rightsquigarrow}\subseteq\mathcal{Z}\times\Sigma_\tau\times\mathcal{Z}$ 
  is the least relation satisfying the rule:\\
  $z\xrightsquigarrow{\mu}z'$ if 
  $\ell\xrightarrow{g,\mu,R}\ell'$, $D'=R(D^\uparrow\land g\land I(\ell))\land I(\ell')$, and
  $\mathcal{H}'=\emph{update}(\mathcal{H},D,D')$.
\end{itemize}
We apply Algorithm~\ref{algorithm:finite-zhg} to generate a \emph{finite} zone-history graph from $(\mathcal{Z},z_0,\Sigma,\rightsquigarrow)$.
\end{definition}

Before we describe Algorithm~\ref{algorithm:finite-zhg} 
for pruning zone-history graphs in more detail (as well as the 
operators used in this algorithm), 
we first provide an example of an infinite zone-history graph.

\begin{example}\label{example:zhg-infinite}
Figures~\ref{fig:reset-history-graph-a} and~\ref{fig:reset-history-graph-a-prime} 
show extracts from the (infinite) zone-history graphs of TA $\mathcal{A}$ and $\mathcal{A}'$, 
respectively (\cf{}~Figs.~\ref{fig:reset-history-a} and~\ref{fig:reset-history-a-prime}), where $\mathcal{A}'$ has two clocks, $y$ and $z$.
The initial state of $\mathcal{ZH}'$ starts in location $\ell'_0$ and zone $y=0\land y=z$.
Considering the switch labeled with $a$, $y\leq2$ and reset of $z$,
we track clock differences in zone-history graphs (\eg{} $y=z$ in the initial state)
as usual, and update difference constraints in case of clock resets~\cite{Dill1989,Weise1997}.
Due to $y\leq2$, the difference between $y$ and $z$ may increase, 
thus resulting in $y\leq z+2$.
The updated zone history yields $\chi\leq2$ with span $(0,2)$.
Next, we update the existing entry of the zone history and append a new entry 
for the current step.
As both $\mathcal{A}$ and $\mathcal{A}'$ contain cycles, we proceed by adding states 
with updated histories, such that the resulting zone-history graphs will become infinite.
\end{example}

We next introduce the auxiliary operators used in Algorithm~\ref{algorithm:finite-zhg}.
By $|\mathcal{H}|$, we denote the \emph{length} of sequence $\mathcal{H}$
and by $\mathcal{H}\downarrow_{k}$, $k>0$, we denote the \emph{postfix} 
of $\mathcal{H}$ of length $k$ (or whole $\mathcal{H}$ if $k\geq |\mathcal{H}|$).
In this way, we compare sequences of spans of two histories of differing lengths 
by only considering a respective postfix of the longer one.
To this end, we compare \emph{spans of histories} 
by comparing the zones of the respective zone histories.
In particular, we use $\mathcal{H}\prec\mathcal{H}'$ and $\mathcal{H}\preceq\mathcal{H}'$ 
to denote an element-by-element comparison of the spans of clock $\chi$ (\ie{} the additional clock introduced in Def.~\ref{def:zone-history}).
Please note that we utilize the generic symbol ${\trianglelefteq}\in\{\prec,\preceq\}$ only 
for the sake of a compact definition.

\begin{definition}[Comparison of Zone Histories]\label{def:zh-comparison}
Let $\mathcal{H},\mathcal{H}'\in\mathcal{B}(C\cup\{\chi\})^*$ with $\chi\notin C$ be \emph{zone histories}.
The \emph{comparison} of the \emph{spans of histories} $\mathcal{H}$ and $\mathcal{H}'$ is recursively defined by
\begin{itemize}
  \item $\mathcal{H}\trianglelefteq\mathcal{H}'$ if $\mathcal{H}=\mathcal{H}'=\epsilon$,
  \item $\mathcal{H}\trianglelefteq\mathcal{H}' \Leftrightarrow \emph{span}(\chi,H)\trianglelefteq\emph{span}(\chi,H') \wedge \mathcal{H}''\trianglelefteq \mathcal{H}'''$ 
  if $|\mathcal{H}|=|\mathcal{H}'| \wedge \mathcal{H}=H\cdot\mathcal{H}'' \wedge \mathcal{H}'=H'\cdot\mathcal{H}'''$, and
  \item $\mathcal{H}\trianglelefteq\mathcal{H}'\Leftrightarrow \mathcal{H}\downarrow_{k}\trianglelefteq\mathcal{H}'\downarrow_{k}$ if 
  $|\mathcal{H}|\neq |\mathcal{H}'|$ and $k=\min(|\mathcal{H}|,|\mathcal{H}'|)$,
\end{itemize}
where ${\trianglelefteq}\in\{\prec,\preceq\}$.
\end{definition}

We illustrate the comparison of zone histories by the following example.

\begin{example}\label{example:zh-comparison}
Consider the following zone histories (where we omit all clocks except $\chi$ for the sake of readability):
\begin{itemize}
  \item $\mathcal{H}=(\chi>4)\cdot(\chi\geq10)\cdot(\chi\geq12\land\chi\leq42)$
  \item $\mathcal{H}'=(\chi>7)\cdot(\chi\geq12\land\chi\leq42)$
  \item $\mathcal{H}''=(\chi\geq17\land\chi\leq41)$
\end{itemize}
Comparing these zone histories, it holds that $\mathcal{H}\preceq\mathcal{H}'$, $\mathcal{H}''\prec\mathcal{H}$, and $\mathcal{H}''\prec\mathcal{H}'$.
In contrast, $\mathcal{H}\preceq\mathcal{H}''$, $\mathcal{H}'\preceq\mathcal{H}''$, and $\mathcal{H}'\preceq\mathcal{H}$ do not hold.
\end{example}

Moreover, we define respective comparison operators on zone histories potentially having
different lengths.
First, $\mathcal{H}\asymp\mathcal{H}'$ compares histories $\mathcal{H}$ and $\mathcal{H}'$
by considering the longest possible postfixes of both zone histories.

\begin{definition}[Postfix-Equality of Zone Histories]\label{def:zh-equality}
Let $\mathcal{H},\mathcal{H}'\in\mathcal{B}(C\cup\{\chi\})^*$ be \emph{zone histories}.
$\mathcal{H}$ and $\mathcal{H}'$ are equal, denoted by $\mathcal{H}\asymp\mathcal{H}'$, iff $\mathcal{H}\preceq\mathcal{H'}$ and $\mathcal{H}'\preceq\mathcal{H}$.
\end{definition}

Second, $\mathcal{H}\asymp_\circlearrowleft\mathcal{H}'$ further \emph{cuts} postfixes in case of periodic zone histories.
The usage of this operator will be explained in more detail later on 
(see Algorithm~\ref{algorithm:finite-zhg}).

\begin{definition}[Cut-Equality of Zone Histories]\label{def:zh-periodic-comparison}
Let $\mathcal{H},\mathcal{H}'\in\mathcal{B}(C\cup\{\chi\})^*$ with $\chi\notin C$ be \emph{zone histories}.
The \emph{periodic comparison} of the \emph{spans of histories} $\mathcal{H}$ and $\mathcal{H}'$ is recursively defined by
\begin{itemize}
  \item $\mathcal{H}\trianglelefteq_\circlearrowleft\mathcal{H}'\Leftrightarrow\mathcal{H}\downarrow_{k}\trianglelefteq\mathcal{H}'\downarrow_{k}$ if $k=\min(|\mathcal{H}|,|\mathcal{H}'|,|\omega|)$ with $\omega=|\mathcal{H}|-|\mathcal{H}'|$ and
  \item $\mathcal{H}\asymp_\circlearrowleft\mathcal{H}'$ if $\mathcal{H}\trianglelefteq_\circlearrowleft\mathcal{H'}$ and $\mathcal{H}'\trianglelefteq_\circlearrowleft\mathcal{H}$,
\end{itemize}
where ${\trianglelefteq}\in\{\prec,\preceq\}$.
\end{definition}

We, again, illustrate the application of these operators by the following example.

\begin{example}\label{example:zh-equality}
Consider, again, the zone histories presented in Example~\ref{example:zh-comparison}.
For instance, $\mathcal{H}\asymp\mathcal{H}'$ does not hold as $\mathcal{H}'\preceq\mathcal{H}$ does not hold.
However, it holds that $\mathcal{H}\asymp_\circlearrowleft\mathcal{H}'$ as this operator only compares a
postfix of length $|\mathcal{H}|-|\mathcal{H}'|=1$ instead of $\mathcal{H}\asymp\mathcal{H}'$ 
which would consider a postfix of length $\min(|\mathcal{H}|,|\mathcal{H}'|)=2$.
\end{example}

Next, we describe Algorithm~\ref{algorithm:finite-zhg} for pruning zone-history graphs.
In particular, the algorithm takes as input a (potentially infinite) zone-history graph 
$(\mathcal{Z},z_0,\Sigma,\rightsquigarrow)$ and returns a finite zone-history graph $(\mathcal{Z}',z_0,\Sigma,\rightsquigarrow')$ with
equivalent behavior. 
To this end, the algorithm stops the unrolling of cyclic behavior based on a cut criterion on zone histories as 
described above.
\begin{algorithm}[tp]
  \caption{Generating Finite Zone-History Graphs}\label{algorithm:finite-zhg}
  \input{figures/algorithm-finite-zhg.tex}
\end{algorithm}
We start by initializing transition relation $\rightsquigarrow'$, the set of 
states $\mathcal{Z}'$, and a working set $\widehat{\mathcal{Z}}$ containing 
states which have not yet been processed (see lines \ref*{line:init-trans}--\ref*{line:init-ws}).
The main loop iterates over this working set $\widehat{\mathcal{Z}}$ until $\widehat{\mathcal{Z}}=\emptyset$ (lines \ref*{line:while-start}--\ref*{line:while-end}).
As a first step of the while-loop, we pick a state $z\in\widehat{\mathcal{Z}}$ 
(without removing it) from the working set (line~\ref*{line:pick-state}).
Then, we check two conditions for each transition $z\xrightsquigarrow{\mu}z'$ 
from transition relation $\rightsquigarrow$ (line~\ref*{line:if-start}).

\begin{enumerate}
  \item Does there already exist some state 
  $\angles{\ell,D,\mathcal{H''}}\in\mathcal{Z}'$ satisfying $\mathcal{H}\asymp_\circlearrowleft\mathcal{H}''$ and $\mathcal{H}\neq\mathcal{H}''$?
  Therewith, we check whether a state $\angles{\ell,D,\mathcal{H''}}\neq z$ has already been reached in a previous step 
  having an equivalent history \wrt{} $\asymp_\circlearrowleft$.
  \item Does there already exist a state $\angles{\ell',D',\mathcal{H}'''}\in\mathcal{Z}'$ 
  satisfying $\mathcal{H}'\asymp_\circlearrowleft\mathcal{H}'''$?
\end{enumerate}

For both properties, we utilize operator $\asymp_\circlearrowleft$ (see Def.~\ref{def:zh-periodic-comparison}) 
for history comparison as this operator only compares the postfix of 
histories reaching back to the last iteration of cyclic behavior (thus \emph{cutting} histories in case of regularity).
If this is the case, we add a transition from $z$ to the previously reached state $\angles{\ell',D',\mathcal{H}'''}$ (line~\ref*{line:add-regular}).
In this way, history unrolling is \emph{cut} whenever states with similar location-zone pairs 
and compatible zone-history postfixes have already been reached before.
Otherwise (lines \ref*{line:else-start}--\ref*{line:else-end}), 
we add the transition and its target state to $\rightsquigarrow'$ and $\mathcal{Z}'$, respectively (lines \ref*{line:add-trans}--\ref*{line:add-state}).
Furthermore, we add the newly explored target state $z'$ to the working set $\widehat{\mathcal{Z}}$ (line~\ref*{line:add-ws}).
Finally, when $\widehat{\mathcal{Z}}=\emptyset$ eventually holds, 
the while-loop terminates and we return the \emph{finite} 
zone-history graph $(\mathcal{Z}',z_0,\Sigma,\rightsquigarrow')$ (line~\ref*{line:return}).
Intuitively, Algorithm~\ref{algorithm:finite-zhg} always eventually terminates 
(which we will prove later in this section) as traversing a loop multiple times always results 
in the exact same postfix in the respective zone-history component
(such that the zone histories are equivalent \wrt{} $\asymp_\circlearrowleft$).

\begin{example}
Consider, again, the TA in Fig.~\ref{fig:reset-history-a} and the corresponding (infinite) 
zone-history graph in Fig.~\ref{fig:reset-history-graph-a} as described in Example~\ref{example:zhg-infinite}.
Here, the target state of the next transition labeled with $c$ does not meet the cut criterion, 
but rather imposes further loop unrolling which results in adding the target state into the zone-history graph.
This newly added state has the same location as the initial state (\ie{} $\ell_0$) as well as 
the same zone (\ie{} $x=0$) due to the reset of clock $x$.
From here, traversing the switch labeled with action $a$ for the second time results in the new history element
$$x\leq2\land\chi\geq0\land\chi=x$$ which is appended to the zone history.
In line~\ref*{line:if-start} of Algorithm~\ref{algorithm:finite-zhg}, 
we then check whether $\mathcal{H}\asymp_\circlearrowleft\mathcal{H}''$ holds.
In this case, $\mathcal{H}$ is the zone history of the current state where we 
just appended the new element as described above.
Furthermore, $\mathcal{H}''$ is the state comprising location $\ell_1$ of 
the zone-history graph in Fig.~\ref{fig:reset-history-graph-a}.
As $\min(|\mathcal{H}|,|\mathcal{H}''|,|\mathcal{H}|-|\mathcal{H}''|)=1$, we only have to compare the postfix 
of length 1 of $\mathcal{H}$ and $\mathcal{H}''$.
As these postfixes (\ie{} the history element described above and the history element depicted in Fig.~\ref{fig:reset-history-graph-a}) 
are equivalent \wrt{} $\asymp_\circlearrowleft$, the condition in line~\ref*{line:if-start} of Algorithm~\ref{algorithm:finite-zhg} 
is satisfied such that we can now cut the zone history and by adding a transition leading back to 
the already existing state thus resulting in a finite zone-history graph.
\end{example}

Next, we formally prove that zone-history graphs resulting from applying Algorithm~\ref{algorithm:finite-zhg} are 
always finite.

\begin{proposition}\label{proposition:algorithm-finite}
Let $\mathcal{A}$ be a TA.
Then, zone-history graph $\mathcal{ZH}_\mathcal{A}$ is finite.
\end{proposition}

\begin{proof}
Zone graphs $(\mathcal{Z},z_0,\Sigma,\rightsquigarrow)$ (without histories) are not necessarily finite but it has been shown that an equivalent finite zone graph $(\mathcal{Z},z_0,\Sigma,\rightsquigarrow_k)$ can be obtained by constructing a \emph{k-bounded} zone graph with all zones being bound by a maximum global \emph{clock ceiling} $k$ using $k$-normalization~\cite{Rokicki1994,Pettersson1999}.
Hence, it remains to be shown that when adding histories, $k$-normalized zone-history graphs $(\mathcal{Z},z_0,\Sigma,\rightsquigarrow_k)$ remain finite.
Here, histories $\mathcal{H}$ are constructed in a way such that $\mathcal{H}$ is eventually cut.
In particular, whenever there already exists a state with the same location $\ell$ and an equivalent zone $D$, we check if $\mathcal{H}\asymp_\circlearrowleft\mathcal{H}'$ and do not add a new state in this case (\ie{} we add a transition to the existing state $\angles{\ell,D,\mathcal{H}}$ instead of adding the new state $\angles{\ell,D,\mathcal{H}'}$, see lines \ref*{line:if-start}--\ref*{line:if-end} of Algorithm~\ref{algorithm:finite-zhg}).
To this end, $\mathcal{H}\asymp_\circlearrowleft\mathcal{H}'$ compares the postfix of $\mathcal{H}$ and $\mathcal{H}'$ of length $n=\min(|\mathcal{H}|,|\mathcal{H}'|,|\omega|)$ with $\omega=|\mathcal{H}|-|\mathcal{H}'|$ (see Def.~\ref{def:zh-periodic-comparison}).
As a result, we only compare the newest $n$ elements of a history when unrolling a loop, where $n$ is the number of locations on the loop.
Therefore, the history eventually becomes regular as we only compare the postfix of length $n$ and TA are \emph{finite} state-transition graphs (see Def.~\ref{def:ta}).
\end{proof}

A proof concerning the correctness of this construction in terms of behavior preservation
will follow later on this section.

\subsection{Composite Zone-History Graphs}

To handle TA with non-deterministic behavior (including $\tau$-steps) 
we require one more concept as illustrated by the following example.

\begin{example}
In Fig.~\ref{fig:zone-split}, we have $\mathcal{A}\simeq\mathcal{A}'$ as both TA permit action 
$a$ within span $(0,3)$ and both switches of $\mathcal{A}'$ labeled with $a$ can be simulated by $\mathcal{A}$.
\begin{figure}[tp]
  \hfill
  \subfloat[$\mathcal{A}$]{\label{fig:zone-split-a}\input{figures/zone-split-a.tex}}
  \hfill
  \subfloat[$\mathcal{A}'$]{\label{fig:zone-split-a-prime}\input{figures/zone-split-a-prime.tex}}
  \hfill
  \subfloat[$\mathcal{ZH}$]{\label{fig:zone-split-graph-a}\input{figures/zone-split-graph-a.tex}}
  \hfill
  \subfloat[$\mathcal{ZH}'$]{\label{fig:zone-split-graph-a-prime}\input{figures/zone-split-graph-a-prime.tex}}
  \hfill\strut
  \caption{Example for State Splitting due to Non-determinism}\label{fig:zone-split}
\end{figure}
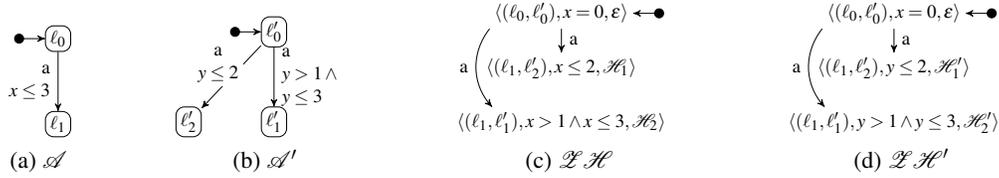
However, the single switch of $\mathcal{A}$ cannot be simulated by either of the two switches of $\mathcal{A}'$.
Hence, generating comparable zone-history graphs for timed-bisimilarity checking may require 
\emph{splitting} of states in case of non-determinism with overlapping spans of guards, as shown in 
Figs.~\ref{fig:zone-split-graph-a} and~\ref{fig:zone-split-graph-a-prime} for $\mathcal{A}$ and $\mathcal{A}'$.
We call this construction \emph{composite} zone-history graph.
\end{example}

The (in general non-symmetric) 
construction of a \emph{composite zone-history graph} $\mathcal{ZH}_{\mathcal{A}\otimes\mathcal{A}'}$ 
for TA $\mathcal{A}$ with respect to $\mathcal{A}'$ is based on the \emph{zone-history graph} $\mathcal{ZH}_{\mathcal{A}\times\mathcal{A}'}$
for the (synchronous) parallel product $\mathcal{A}\times\mathcal{A}'$,
comprising only behavior shared by $\mathcal{A}$ and $\mathcal{A}'$.
Additionally, $\mathcal{ZH}_{\mathcal{A}\otimes\mathcal{A}'}$ also comprises all further behavior
of $\mathcal{ZH}_\mathcal{A}$ potentially not enabled by $\mathcal{ZH}_{\mathcal{A}'}$ 
such that the result is (1) bisimilar to $\mathcal{ZH}_\mathcal{A}$ and (2) facilitates a 
(bi-)simulation check with $\mathcal{ZH}_\mathcal{A}'$ even in the presence of non-deterministic behavior.
In order to construct the composite zone-history graph $\mathcal{ZH}_{\mathcal{A}\otimes\mathcal{A}'}$, 
we first define the parallel product $\mathcal{A}\times\mathcal{A}'$.

\begin{definition}[Parallel Product]\label{def:parallel-product}
Let $\mathcal{A}$, $\mathcal{A}'$ be TA over $\Sigma$ with $C\,\cap\, C' = \emptyset$.
The \emph{parallel product} $\mathcal{A} \times \mathcal{A}' = \left(L \times L',(\ell_0,\ell_0'),\Sigma,C\cup C', I_{\times}, E_{\times}\right)$ 
is a TA with $I_{\times}(\ell,\ell')=I(\ell)\land I(\ell')$ and $E_{\times}$ being the least relation satisfying:
\begin{tabbing}
(1) \= $(\ell_1,\ell_1')\xrightarrow{g\land g',\sigma,R\cup R'}_\times(\ell_2,\ell_2')\in E_\times$ \= if $\ell_1\xrightarrow{g,\sigma,R}\ell_2\in E \land \ell_1'\xrightarrow{g',\sigma,R'}\ell_2'\in E'$, \\
(2) \> $(\ell_1,\ell_1')\xrightarrow{g,\tau,R}_\times(\ell_2,\ell_1')\in E_\times$ \> if $\ell_1\xrightarrow{g,\tau,R}\ell_2\in E$, and \\
(3) \> $(\ell_1,\ell_1')\xrightarrow{g',\tau,R'}_\times(\ell_1,\ell_2')\in E_\times$ \> if $\ell_1'\xrightarrow{g',\tau,R'}\ell_2'\in E'$.
\end{tabbing}
\end{definition}

Next, we introduce two auxiliary transition relations from which we derive the transition 
relation $\rightsquigarrow_\otimes$ of $\mathcal{ZH}_{\mathcal{A}\otimes\mathcal{A}'}$.
Here, $\rightsquigarrow_\times$ denotes the transition relation 
of $\mathcal{ZH}_{\mathcal{A}\times\mathcal{A}'}$ (\ie{} the zone-history 
graph of $\mathcal{A}\times\mathcal{A}'$), whereas $\rightsquigarrow_1$ 
refers to the transition relation of~$\mathcal{ZH}_\mathcal{A}$.

\begin{definition}\label{def:aux-transitions}
Let $\mathcal{A}$ and $\mathcal{A}'$ be TA, $(\mathcal{Z},z_0,\Sigma,\rightsquigarrow)$ be the zone-history graph of $\mathcal{A}$, and $(\mathcal{Z}',z_0',\Sigma',\rightsquigarrow')$ be the zone-history graph of $\mathcal{A}\times\mathcal{A}'$.
By ${\rightsquigarrow_1}={\rightsquigarrow}$ and ${\rightsquigarrow_\times}={\rightsquigarrow'}$
we denote two \emph{auxiliary transition relations} of $\mathcal{A}$ \wrt{} $\mathcal{A}'$.
\end{definition}

As described above, the parallel product only contains 
behavior being common to $\mathcal{A}$ and $\mathcal{A}'$.
In order to ensure that the composite zone-history graph $\mathcal{ZH}_{\mathcal{A}\otimes\mathcal{A}'}$ 
contains the same behavior as the zone-history graph $\mathcal{ZH}_\mathcal{A}$, 
we further have to add behavior of $\mathcal{A}$ not being enabled in $\mathcal{A}'$ to $\mathcal{ZH}_{\mathcal{A}\otimes\mathcal{A}'}$.

\begin{example}\label{example:motivate-join}
Consider, again, the TA depicted in Fig.~\ref{fig:zone-split} and 
let us assume that the switch of $\mathcal{A}'$ labeled with guard $y\geq2$ would be missing.
In this case, the zone-history graph of $\mathcal{A}\times\mathcal{A}'$ does not contain 
all behavior of $\mathcal{A}$, such that we have to add the missing behavior to ensure that 
the composite zone-history graph of $\mathcal{A}$ \wrt{} $\mathcal{A}'$ is 
bisimilar to the zone-history graph of $\mathcal{A}$.
\end{example}

To identify the behavior of zone-history graph $\mathcal{ZH}_\mathcal{A}$ 
already being contained in $\mathcal{ZH}_{\mathcal{A}\otimes\mathcal{A}'}$, 
we employ the notion of a (bi-)simulation relation.
However, as the (timed) behavior of one transition in $\mathcal{ZH}_\mathcal{A}$ may be 
simulated by a combination of multiple transitions in $\mathcal{ZH}_{\mathcal{A}\otimes\mathcal{A}'}$ labeled with the same action, 
we first have to combine the histories of this set of transitions in $\mathcal{ZH}_{\mathcal{A}\otimes\mathcal{A}'}$.
To this end, we define an operator for \emph{joining} histories.
In particular, we compose sets $\mathfrak{H}$ of histories into
a single one in an incremental manner, where two histories $\mathcal{H}$, $\mathcal{H}'$ 
are combined by element-wise disjunction of their components.

In general, disjunction leads to constraints corresponding to non-convex polyhedra 
(as opposed to convex polyhedra obtained by clock constraints described in Def.~\ref{def:ta}).
As comparing non-convex polyhedra (\eg{} checking if two polyhedra intersect) is less efficient than comparing convex polyhedra,
the construction of a composite zone-history graph is computationally much more complex than constructing zone-history graphs
for deterministic TA.
In fact, those non-convex constraints solely occur during those particular checks determining 
whether additional states must be added.
In contrast, all other constraints emerging during zone-history graph construction,
including those specifying the individual components of histories, always remain convex.
In the following, we first consider the case where $|\mathcal{H}|=|\mathcal{H}'|$.

\begin{definition}[History Join]\label{def:history-join}
Let $\mathfrak{H}\in2^{\mathcal{B}(C)^*}$ be a set of histories.
Function $\join:2^{\mathcal{B}(C)^*}\rightarrow\mathcal{B}(C)^*$ is recursively defined by
\begin{itemize}
  \item $\join(\emptyset)=\epsilon$,
  \item $\join(\{\mathcal{H}\}\cup\mathfrak{H})=\mathcal{H}\dot\lor\join(\mathfrak{H})$,
  \item $H\dot\lor\epsilon=H$, and
  \item $(H\cdot\mathcal{H})\dot\lor(H'\cdot\mathcal{H}')=(H\lor H')\cdot(\mathcal{H}\dot\lor\mathcal{H}')$ if $|\mathcal{H}|=|\mathcal{H}'|$.
\end{itemize}
\end{definition}

In order to join two histories $\mathcal{H}$ and $\mathcal{H}'$ of different length
(\ie{} $|\mathcal{H}|\neq|\mathcal{H}'|$), we expand the shorter history to length
$k=\max(|\mathcal{H}|,|\mathcal{H}'|)$.
To this end, we use the notation $\mathcal{H}\uparrow^k$ to add 
constant elements $\mathsf{false}\in\mathcal{B}(C)$, being the neutral element of disjunction, 
as additional prefixes to $\mathcal{H}$ until $\mathcal{H}$ has length $k$.

\begin{definition}\label{def:history-join-unequal-length}
Let $\mathcal{H},\mathcal{H}'\in\mathcal{B}(C\cup\{\chi\})^*$ be zone histories.
\begin{itemize}
  \item $\mathcal{H}\dot\lor\mathcal{H}'\Leftrightarrow\mathcal{H}\uparrow^k\dot\lor\mathcal{H}'\uparrow^k$ if $|\mathcal{H}|\neq|\mathcal{H}'|$ and $k=\max(|\mathcal{H}|,|\mathcal{H}'|)$,
  \item $\mathcal{H}\uparrow^k=\mathcal{H}$ if $|\mathcal{H}|\geq k$, and
  \item $\mathcal{H}\uparrow^k=\mathsf{false}\cdot\mathcal{H}\uparrow^{k-1}$ if $|\mathcal{H}|<k$.
\end{itemize}
\end{definition}

Next, we define the function \emph{histories} to define composite 
zone-history graphs in a compact way.
Function \emph{histories} takes as input symbolic state $z$, 
symbolic transition relation $\rightsquigarrow$, and action $\mu\in\Sigma_\tau$
and returns the histories of all states being reachable from 
$z$ under action $\mu$.

\begin{definition}\label{def:history-aux}
Let $z\in\mathcal{Z}=L\times\mathcal{B}(C)\times\mathcal{B}(C)^*$ 
be a symbolic state, ${\rightsquigarrow}\in\mathcal{Z}\times\Sigma_\tau\times\mathcal{Z}$ be a symbolic transition relation, and $\mu\in\Sigma_\tau$ be an action.
Function
$$\histories:\mathcal{Z}\times(\mathcal{Z}\times\Sigma_\tau\times\mathcal{Z})\times\Sigma_\tau\rightarrow2^{\mathcal{B}(C)^*}$$
denotes the set of histories $\mathfrak{H}\in2^{\mathcal{B}(C)^*}$ 
being reachable from $z$ with $\mu$, such that $\mathcal{H}'\in\mathfrak{H}$ if $z\xrightsquigarrow{\mu}\angles{\ell',D',\mathcal{H}'}$.
\end{definition}

We are now able to define \emph{composite zone-history graphs} by considering 
the aforementioned transition relations $\rightsquigarrow_1$ and $\rightsquigarrow_\times$ (see Def.~\ref{def:aux-transitions}) 
for TA $\mathcal{A}$ and $\mathcal{A}'$.
For transition relation $\rightsquigarrow_\otimes$ of the composite zone-history 
graph $\mathcal{ZH}_{\mathcal{A}\otimes\mathcal{A}'}$, we require $\rightsquigarrow_\times\subseteq\rightsquigarrow_\otimes$.
In addition, transition $\angles{(\ell_1,\ell_1'),D_1,\mathcal{H}_1}\xrightsquigarrow{\mu}_1\angles{(\ell_2,\ell_1'),D_2,\mathcal{H}_2}$ is also part of $\rightsquigarrow_\otimes$ if the (timed) behavior of this transition is not covered by a transition (or a combination of transitions) in $\rightsquigarrow_\times$ (see Example~\ref{example:motivate-join}).
Hence, $\rightsquigarrow_\otimes$ contains the behavior of the parallel product \emph{and} the behavior exclusive to $\mathcal{A}$.

\begin{definition}[Composite Zone-History Graph]\label{def:composite-zone-history-graph}
Let $\mathcal{A}$, $\mathcal{A}'$ be TA over $\Sigma$ with $C\,\cap\, C' = \emptyset$, $\mathcal{A} \times \mathcal{A}' = \left(L \times L',(\ell_0,\ell_0'),\Sigma,C\cup C', I_{\times}, E_{\times}\right)$ be the parallel product (see Def.~\ref{def:parallel-product}), and $\rightsquigarrow_\times$ and $\rightsquigarrow_1$ be auxiliary transition relations (see Def.~\ref{def:aux-transitions}).
The \emph{composite zone-history graph} $\mathcal{ZH}_{\mathcal{A}\otimes\mathcal{A}'}=(\mathcal{Z},z_0,\Sigma,\rightsquigarrow_{\otimes})$
of $\mathcal{A}$ \wrt{} $\mathcal{A}'$ is a zone-history graph, where
\begin{itemize}
  \item $\mathcal{Z}=(L\times L')\times\mathcal{B}(C\cup C')\times\mathcal{B}(C\cup C'\cup\{\chi\})^*$ is a set of \emph{symbolic states} 
  with \emph{initial state} $z_0=\angles{(\ell_0,\ell_0'),D_0,\epsilon}\in \mathcal{Z}$,
  \item $\Sigma$ is a set of actions and
  \item ${\rightsquigarrow_{\otimes}}\subseteq\mathcal{Z}\times\Sigma_\tau\times\mathcal{Z}$ is the least relation satisfying
  \begin{itemize}
    \item $z\xrightsquigarrow{\mu}_\otimes z'$ if $z\xrightsquigarrow{\mu}_\times z'$ and
    \item $z_1\xrightsquigarrow{\mu}_\otimes\angles{(\ell_2,\ell_1'),D_2,\mathcal{H}_2}$ if $z_1\xrightsquigarrow{\mu}_1\angles{(\ell_2,\ell_1'),D_2,\mathcal{H}_2}\land\join(\mathfrak{H})\prec\mathcal{H}_2$, where $\mathfrak{H}=\histories(z_1,\rightsquigarrow_\times,\mu)$.
  \end{itemize}
\end{itemize}
\end{definition}

This construction allows us to establish a \emph{symbolic version of 
(strong) timed (bi-)\allowbreak simulation} on zone-history graphs such that state $z_1'$ simulates
state $z_1$ if (1) $z_1'$ enables the same actions $\mu\in\Sigma_\tau$ as $z_1$, and
(2) the span of history $\mathcal{H}_1'$ includes the span of $\mathcal{H}_1$, respectively.
Moreover, we have to compare the spans allowed for residing in related states.
As before, we perform this check by introducing a fresh clock $\chi$ and checking the span of $\chi$.
As composite zone-history graphs are, by construction, proper zone-history 
graphs, the following definitions and results are likewise applicable.

\begin{definition}[Symbolic Timed Bisimulation]\label{def:symbolic-timed-bisim}
Let $\mathcal{A}$, $\mathcal{A}'$ be TA over $\Sigma$ 
with $C\cap C'=\emptyset$, $\chi,\chi'\notin  C\cup C'$, symbolic states $\mathcal{Z}$, $\mathcal{Z'}$, and $\mathcal{R}\subseteq\mathcal{Z}\times\mathcal{Z}'$ such that for all $(z_1,z_1')\in\mathcal{R}$
\begin{itemize}
  \item if $z_1\xrightsquigarrow{\mu}z_2$ with $\mu\in\Sigma_\tau$, then $z_1'\xrightsquigarrow{\mu}z_2'$ and $(z_2,z_2')\in\mathcal{R}$ and
  \item $\textit{span}(\chi,(D_1\land\chi=0)^\uparrow\land I(\ell_1))\leq\textit{span}(\chi',(D_1'\land\chi'=0)^\uparrow\land I'(\ell_1'))\land\mathcal{H}_1\preceq \mathcal{H}_1'$.
\end{itemize}
$\mathcal{A}'$ \emph{(strongly) timed simulates} $\mathcal{A}$ iff $(z_0,z_0')\in\mathcal{R}$.
$\mathcal{A}'$ and $\mathcal{A}$ are \emph{(strongly) timed bisimilar}, denoted $\mathcal{A}\simeq\mathcal{A}'$, iff $\mathcal{R}$ is symmetric.
\end{definition}

We overload $\sqsubseteq$ and $\simeq$ on zone-history graphs, accordingly,
and we, again, obtain \emph{weak} versions of those definition as before.
Concerning correctness and decidability of symbolic timed bisimulation on zone-history graphs, we first prove
that the composite zone-history graph is \emph{semantic-preserving} and \emph{finite}.

\begin{proposition}\label{proposition:composite-correctness}
Let $\mathcal{A}$, $\mathcal{A}'$ be TA over $\Sigma$.
Then it holds that (1)~$\mathcal{ZH}_{\mathcal{A}}\simeq\mathcal{ZH}_{\mathcal{A}\otimes\mathcal{A}'}$, 
and (2)~$\mathcal{ZH}_\mathcal{A}$ and $\mathcal{ZH}_{\mathcal{A}\otimes\mathcal{A}'}$ are finite.
\end{proposition}

\begin{proof}
Let $\mathcal{A}$ and $\mathcal{A}'$ be TA over $\Sigma$.
We prove the two parts of Proposition~\ref{proposition:composite-correctness} separately.
\begin{enumerate}
\item By definition, $\mathcal{ZH}_{\mathcal{A}\times\mathcal{A}'}$ contains exactly the shared behavior of $\mathcal{ZH}_\mathcal{A}$ and $\mathcal{ZH}_{\mathcal{A}'}$ (see relation $\rightsquigarrow_\times$ of Def.~\ref{def:aux-transitions}).
In addition, the remaining behavior being exclusive to $\mathcal{ZH}_\mathcal{A}$ is added by relation $\rightsquigarrow_1$ as $\rightsquigarrow_1$ contains exactly the behavior of $\mathcal{A}$.
Furthermore, the requirement $\join(\mathfrak{H})\prec\mathcal{H}_2$ ensures that transitions of $\rightsquigarrow_1$ are added to $\rightsquigarrow_\otimes$ if and only if the respective behavior is not already contained in $\rightsquigarrow_\otimes$ through $\rightsquigarrow_\times$ (see Def.~\ref{def:composite-zone-history-graph}).
Hence, it directly follows that $\mathcal{ZH}_\mathcal{A}\simeq\mathcal{ZH}_{\mathcal{A}\times\mathcal{A}'}$ due to $\mathcal{A}\simeq\mathcal{A}'$ as shown with bisimilarity of the corresponding TLTS.
\item Finiteness of $\mathcal{ZH}_\mathcal{A}$ has already been proven in Proposition~\ref{proposition:algorithm-finite}.
Hence, it remains to be shown that also $\mathcal{ZH}_{\mathcal{A}\otimes\mathcal{A}'}$ is finite.
For the construction of $\mathcal{ZH}_{\mathcal{A}\otimes\mathcal{A}'}$, we first generate the zone-history graphs of $\mathcal{A}$ as well as $\mathcal{A}\times\mathcal{A}'$, which are finite (see Proposition~\ref{proposition:algorithm-finite}).
To obtain $\mathcal{ZH}_{\mathcal{A}\otimes\mathcal{A}'}$, we then add transitions from $\mathcal{ZH}_\mathcal{A}$ to $\mathcal{ZH}_{\mathcal{A}\otimes\mathcal{A}'}$ iff behavior of $\mathcal{A}$ is uncovered.
As the zone-history graphs of $\mathcal{A}$ as well as $\mathcal{A}\times\mathcal{A}'$ are finite, also $\mathcal{ZH}_{\mathcal{A}\otimes\mathcal{A}'}$ is finite.
\qedhere
\end{enumerate}
\end{proof}

Thereupon, we are now able to show correctness of symbolic timed (bi-)simulation.

\begin{theorem}\label{theorem:hzg-bisim-correctness}
Let $\mathcal{A}$, $\mathcal{A}'$ be TA over $\Sigma$.
Then it holds that (1) $\mathcal{A}\sqsubseteq\mathcal{A}' \Leftrightarrow 
\mathcal{ZH}_{\mathcal{A}\otimes\mathcal{A}'} \sqsubseteq \mathcal{ZH}_{\mathcal{A'}\otimes\mathcal{A}}$,
and (2) $\mathcal{ZH}_{\mathcal{A}\otimes\mathcal{A}'} \sqsubseteq \mathcal{ZH}_{\mathcal{A'}\otimes\mathcal{A}}$
is decidable.
\end{theorem}

\begin{proof}
Let $\mathcal{A}$, $\mathcal{A}'$ be TA over $\Sigma$.
We prove the two parts of Theorem~\ref{theorem:hzg-bisim-correctness} separately.
\begin{enumerate}
\item It holds, by construction of composite zone-history graphs, that $\rightsquigarrow_\times=\rightsquigarrow_\times'$ up to renaming of locations and clocks (see Def.~\ref{def:composite-zone-history-graph}).
Hence, \Wlog{}, we have to show that behavior in $\rightsquigarrow_1$ (\ie{} being exclusive to $\mathcal{ZH}_{\mathcal{A}\times\mathcal{A}'}$) cannot be simulated by $\mathcal{ZH}_{\mathcal{A}'\times\mathcal{A}}$.
This follows directly from the first condition of Def.~\ref{def:symbolic-timed-bisim} and the fact that transitions are added to $\rightsquigarrow_1$ iff the corresponding behavior is exclusive (see second rule for $\rightsquigarrow_\otimes$ in Def.~\ref{def:composite-zone-history-graph}).
Furthermore, exclusive behavior of $\mathcal{A}$ cannot be simulated by $\mathcal{A}'$ when considering timed bisimulation on TLTS (see Defs.~\ref{def:ta}, \ref{def:tlts-semantics}, and~\ref{def:ta-bisim}).
Finally, we have to consider that clock resets hide clock constraints in the sense that a clock constraint $x\sim n$ is not visible in a zone after $x$ is reset.
However, by comparing zone histories $\mathcal{H}$ and $\mathcal{H}'$, we ensure that the impact of previous clock constraints remains observable by using the fresh clock $\chi$ for tracking respective changes to clock differences including those potentially being hidden by subsequent clock resets.
Therefore, it holds that $\mathcal{A}\sqsubseteq\mathcal{A}' \Leftrightarrow \mathcal{ZH}_{\mathcal{A}\times\mathcal{A}'} \sqsubseteq \mathcal{ZH}_{\mathcal{A'}\times\mathcal{A}}$.
Note, that $k$-normalization does not impact the bisimilarity check as checking bisimilarity relies on the comparison of histories.
In particular, loops (being the reason for $k$-normalization) result in the comparison of the postfix of length $n=\min(|\mathcal{H}|,|\mathcal{H}'|,|\omega|)$ with $\omega=|\mathcal{H}|-|\mathcal{H}'|$ of histories $\mathcal{H}$ and $\mathcal{H}'$ (see Def.~\ref{def:zone-history}).
As a result, we only compare the newest $n$ elements of a history when unrolling a loop, where $n$ is the number of locations on the loop.
Therefore, the history eventually becomes regular as we only compare the postfix of length $n$, such that we do not apply any approximation to histories.
\item As composite zone-history graphs are finite (see Proposition~\ref{proposition:composite-correctness}), there are finitely many transitions and spans to check (see Def.~\ref{def:symbolic-timed-bisim}).
Hence, $\mathcal{ZH}_{\mathcal{A}\times\mathcal{A}'} \sqsubseteq \mathcal{ZH}_{\mathcal{A'}\times\mathcal{A}}$ is decidable.
\qedhere
\end{enumerate}
\end{proof}

\begin{example}\label{example:zh-graph}
The extract from the zone-history graphs in Fig.~\ref{fig:zone-graph-example} 
correspond to the TA in Fig.~\ref{fig:ta-example}.
Starting from the initial state of \emph{coffee machine} 
(\cf{}~Fig.~\ref{fig:zone-graph-example-machine}) with zone $x=0$, the zone of the
subsequent state is $x=0$ due to the reset, whereas 
the following state has zone $x=1$ due to the invariant of location \emph{Warm Up} and the guard of switch \emph{coffee}.
Additionally, $\mathcal{H}_1=(x=0\land\chi\geq x)$ as $x$ is reset, and 
$\mathcal{H}_2=(x=1\land\chi\geq x)\cdot(x=1\land\chi\geq x)$ due to the guard and invariant.
All elements of $\mathcal{H}_1'$ and $\mathcal{H}_2'$ equal $(y=0\land\chi\geq y)$
while $\mathcal{H}_3'=(y\geq0\land\chi\geq y)\cdot(y\geq0\land\chi\geq y)\cdot(y\geq0\land\chi=y)$.
Hence, both TA are not \emph{strongly} but \emph{weakly} bisimilar as, \eg{} $\mathcal{H}_1\preceq\mathcal{H}_1'$ 
and $\mathcal{H}_1'\preceq\mathcal{H}_1$ (as $\emph{span}(\chi,x=0\land\chi\geq x)=\emph{span}(\chi,y=0\land\chi\geq y)=(0,\infty)$).
Furthermore, TA in Fig.~\ref{fig:zone-graph-example-machine-prime} 
may immediately produce \emph{sugar} after action \emph{coffee} due to silent steps.
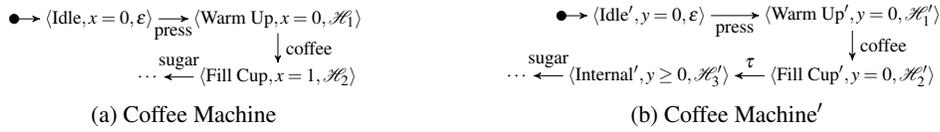
\begin{figure}[tp]
  \hfill
  \subfloat[Coffee Machine]{\label{fig:zone-graph-example-machine}\input{figures/zone-graph-example-machine.tex}}
  \hfill
  \subfloat[Coffee Machine$'$]{\label{fig:zone-graph-example-machine-prime}\input{figures/zone-graph-example-machine-prime.tex}}
  \hfill\strut
  \caption{Zone-History Graphs for TA Depicted in Fig.~\ref{fig:ta-example}}\label{fig:zone-graph-example}
\end{figure}
\end{example}

As shown in Proposition~\ref{proposition:composite-correctness}, zone-history
graphs are finite and allow for precise checking of timed bisimilarity.
However, in case of larger TA models with many locations and clocks, complex clock constraints 
and frequent clock resets, zone-histories graphs may become very large thus obstructing
effective timed bisimilarity-checking by practical tools.
To also handle realistic models, we next define 
bounded zone-history graphs to enable potentially imprecise, yet arbitrarily scalable
timed bisimilarity-checking.

\subsection{Bounded Zone-History Graphs}

For controlling the size of zone-history graphs,
we introduce a \emph{bound parameter} $b\in \mathbb{N}_{0}$
restricting each history sequence $\mathcal{H}$
produced by the \emph{update}-operator (Def.~\ref{def:zone-history})
during zone-history graph construction to $\mathcal{H}\downarrow_{b}$
(\ie{} memorizing a maximum number of $b$ previous history elements).
By $\mathcal{A}\simeq_{b}\mathcal{A}'$, we denote that 
the $b$-bounded zone-history graphs of TA $\mathcal{A}$ and $\mathcal{A}'$
are timed bisimilar. 
Hence, $\mathcal{A}\simeq_{\infty}\mathcal{A}'$
denotes the unbounded case being equivalent to $\mathcal{A}\simeq\mathcal{A}'$, whereas
$\mathcal{A}\simeq_{0}\mathcal{A}'$ denotes timed bisimilarity-checking
on plain zone graphs according to Def.~\ref{def:zone-graph}.

\begin{theorem}\label{theorem:bzhg-bisim}
Let $\mathcal{A},\mathcal{A}'$ be TA over $\Sigma$.
\begin{enumerate}
  \item There exists $b<\infty$ such that $\mathcal{A}\simeq_{b}\mathcal{A}'\Leftrightarrow\mathcal{A}\simeq\mathcal{A}'$.
  \item $\mathcal{A}\simeq_{b}\mathcal{A}'\Rightarrow\mathcal{A}\simeq_{b'}\mathcal{A}'$ iff $b\geq b'$.
\end{enumerate}
\end{theorem}

\begin{proof}
We prove (1) and (2) separately.
\begin{enumerate}
\item As $\mathcal{A}\simeq\mathcal{A}'$ is decidable (see Theorem~\ref{theorem:hzg-bisim-correctness}) and (composite) zone-history graphs have a finite length (see Proposition~\ref{proposition:composite-correctness}), the length of the respective zone history is finite.
Hence, there exists $b<\infty$ where $b$ may have the length of the longest zone history when computing $\mathcal{A}\simeq\mathcal{A}'$.

\item If it holds that $\mathcal{A}\simeq_{b}\mathcal{A}'$, then it also holds that $\mathcal{A}\simeq_{b'}\mathcal{A}'$ iff $b\geq b'$ as $\mathcal{A}\simeq_{b'}\mathcal{A}'$ considers a shorter history (where the leading elements of the history are equal to considering $b$).
Here, recognizing a TA $\mathcal{A}'$ as not bisimilar would require an element in the zone history to be unequal.
\qedhere
\end{enumerate}
\end{proof}

However, identifying a minimal, yet sufficiently large $b$ meeting the first property a-priori is
not obvious.
In contrast, if $\mathcal{A}\not\simeq_{b}\mathcal{A}'$ holds for some $b$, 
then $\mathcal{A}\simeq\mathcal{A}'$ does also not hold, whereas
$\mathcal{A}\simeq_{b}\mathcal{A}'$ may be false positive only if
histories exceed bound $b$ at least once during zone-history-graph construction.

\begin{example}
Let us assume $b=1$ in Fig.~\ref{fig:reset-history}.
Here, the history contains the constraint $\chi\leq2$ for the states of $\mathcal{ZH}$ and $\mathcal{ZH}'$ comprising $\ell_1$ and $\ell_1'$, respectively.
In states containing $\ell_2$ and $\ell_2'$, respectively, we have $\chi\leq5$ on both sides 
as we only consider the tailing history elements due to $b=1$.
Hence, $\mathcal{A}\simeq_{1}\mathcal{A}'$.
In contrast, $b\geq2$ yields the correct result $\mathcal{A}\not\simeq_{b}\mathcal{A}'$ as we also 
consider the differing first history elements $\chi\leq5$ and $\chi\leq7$ of the states containing $\ell_2$ and $\ell_2'$
thus revealing the effect of the reset of $z$ in $\mathcal{A}'$.
\end{example}

%% file: figures/reset-history-a.tex

\scalebox{\scalefactor}{
\begin{tikzpicture}[node distance=.9]

\node[labeledroundedstate, initial, initial where=left] (ell0) {$\ell_0$};
\node[labeledroundedstate, below=of ell0] (ell1) {$\ell_1$};
\node[labeledroundedstate, right=of ell1] (ell2) {$\ell_2$};

\draw[transition] (ell0) to node[auto, swap, align=right]{a\\$x\leq2$} (ell1);
\draw[transition] (ell1) to node[pos=.5, align=center]{b\\$x\leq5$} (ell2);
\draw[transition] (ell2) to node[auto, swap, align=center]{c\\$x:=0$} (ell0);

\end{tikzpicture}
}

%% file: figures/reset-history-a-prime.tex

\scalebox{\scalefactor}{
\begin{tikzpicture}[node distance=.9]

\node[labeledroundedstate, initial, initial where=left] (ell0) {$\ell_0'$};
\node[labeledroundedstate, below=of ell0] (ell1) {$\ell_1'$};
\node[labeledroundedstate, right=of ell1] (ell2) {$\ell_2'$};

\draw[transition] (ell0) to node[auto, swap, align=right]{a\\$y\leq2$\\$z:=0$} (ell1);
\draw[transition] (ell1) to node[pos=.5, align=center]{b\\$z\leq5$} (ell2);
\draw[transition] (ell2) to node[auto, swap, align=center]{c\\$y:=0$\\$z:=0$} (ell0);

\end{tikzpicture}
}

%% file: figures/reset-history-graph-a.tex

\scalebox{\scalefactor}{
\begin{tikzpicture}[node distance=.4]

\node[labeledstate, initial, initial where=right] (ell0) {$\angles{\ell_0,x=0,\epsilon}$};
\node[labeledstate, below=of ell0] (ell1) {$\angles{\ell_1,x\leq2,\mathstrut$\\$(x\leq 2\land\chi\geq0\land\chi=x)}$};
\node[labeledstate, below=of ell1] (ell2) {$\langle\ell_2,x\leq5,$\\$(x\leq5\land\chi\geq0\land\chi=x)\cdot\mathstrut$\\$(x\leq5\land\chi\geq0\land\chi=x)\rangle$};
\node[labeledstate, below=of ell2] (empty) {$\cdots$};

\draw[transition] (ell0) to node[auto]{a} (ell1);
\draw[transition] (ell1) to node[auto]{b} (ell2);
\draw[transition] (ell2) to node[auto]{c} (empty);

\end{tikzpicture}
}

%% file: figures/reset-history-graph-a-prime.tex

\scalebox{\scalefactor}{
\begin{tikzpicture}[node distance=.4]

\node[labeledstate, initial, initial where=right] (ell0) {$\angles{\ell_0',y=0\land y=z,\epsilon}$};
\node[labeledstate, below=of ell0, align=center] (ell1) {$\angles{\ell_1',y\leq2\land z=0\land y\leq z+2,$\\$(y\leq2\land z=0\land y\leq z+2\land\strut$\\$\chi\geq0\land\chi=y\land\chi\leq z+2)}$};
\node[labeledstate, below=of ell1, align=center] (ell2) {$\langle\ell_2',y\geq0\land z\leq5\land y\leq z+2,$\\$(y\geq0\land z\leq5\land y\leq z+2\land\chi\geq0\land\chi=y\land\chi\leq z+2)\cdot\mathstrut$\\$(y\geq0\land z\leq5\land y\leq z+2\land\chi\geq0\land\chi\leq y\land\chi=z)\rangle$};
\node[labeledstate, below=of ell2] (empty) {$\cdots$};

\draw[transition] (ell0) to node[auto]{a} (ell1);
\draw[transition] (ell1) to node[auto]{b} (ell2);
\draw[transition] (ell2) to node[auto]{c} (empty);

\end{tikzpicture}
}

%% file: figures/algorithm-finite-zhg.tex
 
\SetKwInOut{Input}{Input}
\SetKwInOut{Output}{Output}

\SetKwProg{procedure}{procedure}{}{}

\SetKwFunction{main}{\normalfont\scshape main}
\SetKwFunction{getAllValueSchemas}{\normalfont\scshape getAllValueSchemas}
\SetKwFunction{getTestCase}{\normalfont\scshape getTestCase}

\DontPrintSemicolon

\Input{zone-history graph $(\mathcal{Z},z_0,\Sigma,\rightsquigarrow)$}
\Output{\emph{finite} zone-history graph $(\mathcal{Z}',z_0,\Sigma,\rightsquigarrow')$}

\procedure{\main}{
	${\rightsquigarrow'}:=\emptyset$\; \label{line:init-trans}
	$\mathcal{Z}':=\{z_0\}$\;
	$\widehat{\mathcal{Z}}:=\{z_0\}$\; \label{line:init-ws}
	\While{$\widehat{\mathcal{Z}}\neq\emptyset$}{ \label{line:while-start}
		$z\leftarrow\widehat{\mathcal{Z}}$ \tcp*[r]{pick element without removing it} \label{line:pick-state}
		\ForEach{$z\xrightsquigarrow{\mu}z'$}{
			\If{$\exists\angles{\ell,D,\mathcal{H''}}\in\mathcal{Z}':\left(\mathcal{H}\asymp_\circlearrowleft\mathcal{H}''\land\mathcal{H}\neq\mathcal{H}''\right)\land\mathstrut$\linebreak$\exists\angles{\ell',D',\mathcal{H}'''}\in\mathcal{Z}':\mathcal{H}'''\asymp_\circlearrowleft\mathcal{H}'$}{ \label{line:if-start}
				${\rightsquigarrow'}:={\rightsquigarrow'}\cup\{z\xrightsquigarrow{\mu}\mathstrut'\angles{\ell',D',\mathcal{H}'''}\}$\; \label{line:add-regular}
			} \label{line:if-end}
			\Else{ \label{line:else-start}
				${\rightsquigarrow'}:={\rightsquigarrow'}\cup\{z\xrightsquigarrow{\mu}\mathstrut'z'\}$\; \label{line:add-trans}
				$\mathcal{Z}':=\mathcal{Z}'\cup\{z'\}$\; \label{line:add-state}
				$\widehat{\mathcal{Z}}:=\widehat{\mathcal{Z}}\cup\{z'\}$\; \label{line:add-ws}
			} \label{line:else-end}
		}
		$\widehat{\mathcal{Z}}:=\widehat{\mathcal{Z}}\setminus\{z\}$\;
	} \label{line:while-end}
	\Return{$(\mathcal{Z}',z_0,\Sigma,\rightsquigarrow')$} \label{line:return}
}

%% file: figures/zone-split-a.tex

\scalebox{\scalefactor}{
\begin{tikzpicture}[node distance=1]

\node[labeledroundedstate, initial, initial where=left] (ell0) {$\ell_0$};
\node[labeledroundedstate, below=of ell0] (ell1) {$\ell_1$};

\draw[transition] (ell0) to node[auto, swap, align=right]{a\\$x\leq3$} (ell1);

\end{tikzpicture}
}

%% file: figures/zone-split-a-prime.tex

\scalebox{\scalefactor}{
\begin{tikzpicture}[node distance=1]

\node[labeledroundedstate, initial, initial where=left] (ell0) {$\ell_0'$};
\node[labeledroundedstate, below=of ell0] (ell1) {$\ell_1'$};
\node[labeledroundedstate, left=of ell1] (ell2) {$\ell_2'$};

\draw[transition] (ell0) to node[auto, align=left]{a\\$y>1\land\mathstrut$\\$y\leq3$} (ell1);
\draw[transition] (ell0) to node[pos=.3, left, descr, align=center]{a\\$y\leq2$} (ell2);

\end{tikzpicture}
}

%% file: figures/zone-split-graph-a.tex

\scalebox{\scalefactor}{
\begin{tikzpicture}[node distance=.4]

\node[labeledstate, initial, initial where=right] (ell0) {$\angles{(\ell_0,\ell_0'),x=0,\epsilon}$};
\node[labeledstate, below=of ell0] (ell1) {$\angles{(\ell_1,\ell_2'),x\leq2,\mathcal{H}_1}$};
\node[labeledstate, below=of ell1] (ell2) {$\angles{(\ell_1,\ell_1'),x>1\land x\leq3,\mathcal{H}_2}$};

\draw[transition] (ell0) to node[auto]{a} (ell1);
\draw[transition, bend right=40] (ell0.south west) to node[auto, swap]{a} (ell2);

\end{tikzpicture}
}

%% file: figures/zone-split-graph-a-prime.tex

\scalebox{\scalefactor}{
\begin{tikzpicture}[node distance=.4]

\node[labeledstate, initial, initial where=right] (ell0) {$\angles{(\ell_0,\ell_0'),x=0,\epsilon}$};
\node[labeledstate, below=of ell0] (ell1) {$\angles{(\ell_1,\ell_2'),y\leq2,\mathcal{H}_1'}$};
\node[labeledstate, below=of ell1] (ell2) {$\angles{(\ell_1,\ell_1'),y>1\land y\leq3,\mathcal{H}_2'}$};

\draw[transition] (ell0) to node[auto]{a} (ell1);
\draw[transition, bend right=40] (ell0.south west) to node[auto, swap]{a} (ell2);

\end{tikzpicture}
}

%% file: figures/zone-graph-example-machine.tex

\scalebox{\scalefactor}{
\begin{tikzpicture}[node distance=.5]

\node[labeledstate, initial, initial where=left] (idlezero) {$\angles{\text{Idle},x=0,\epsilon}$};
\node[labeledstate, right=of idlezero] (warmupzero) {$\angles{\text{Warm Up},x=0,\mathcal{H}_1}$};
\node[labeledstate, below=of warmupzero] (fillcupone) {$\angles{\text{Fill Cup},x=1,\mathcal{H}_2}$};
\node[labeledstate, left=of fillcupone] (empty) {$\cdots$};

\draw[transition] (idlezero) to node[auto, swap]{press} (warmupzero);
\draw[transition] (warmupzero) to node[auto]{coffee} (fillcupone);
\draw[transition] (fillcupone) to node[auto, swap]{sugar} (empty);

\end{tikzpicture}
}

%% file: figures/zone-graph-example-machine-prime.tex

\scalebox{\scalefactor}{
\begin{tikzpicture}[node distance=.5]

\node[labeledstate] (warmupzero) {$\angles{\text{Warm Up}',y=0,\mathcal{H}_1'}$};
\node[labeledstate, below=of warmupzero] (fillcupzero) {$\angles{\text{Fill Cup}',y=0,\mathcal{H}_2'}$};
\node[labeledstate, left=of fillcupzero] (internalzero) {$\angles{\text{Internal}',y\geq0,\mathcal{H}_3'}$};
\node[labeledstate, above=of internalzero, initial, initial where=left] (idlezero) {$\angles{\text{Idle}',y=0,\epsilon}$};
\node[labeledstate, left=of internalzero] (empty) {$\cdots$};

\draw[transition] (idlezero) to node[auto, swap]{press} (warmupzero);
\draw[transition] (warmupzero) to node[auto]{coffee} (fillcupzero);
\draw[transition] (fillcupzero) to node[auto, swap]{$\tau$} (internalzero);
\draw[transition] (internalzero) to node[auto, swap]{sugar} (empty);

\end{tikzpicture}
}

%% file: sections/implementation.tex

\section{Implementation}\label{section:implementation}

We implemented the concepts for checking (weak and strong) timed bisimilarity
as described in the previous section which we will describe
in more detail in the following.

Our tool is called \timbrcheck{} (\textbf{tim}ed \textbf{b}isimila\textbf{r}ity \textbf{check}er) and
uses \uppaal{}~\cite{Larsen1997}, a widely used 
tool environment for TA modeling and analysis, as a front-end.
To this end, \timbrcheck{} supports the \uppaal{} file format for input TA models.
After parsing two given input TA models, our tool generates (bounded) zone-history 
graphs for a predefined bound value $b$ and performs a timed-bisimilarity check 
between both models. 
Our tool also supports input models having non-deterministic behavior as well as $\tau$-transitions 
by constructing the corresponding composite (bounded) zone-history graphs as described above.

Internally, \timbrcheck{} utilizes difference bound matrices 
(DBM)~\cite{Bellman1957,Dill1989,Bengtsson2003} 
as a common data structure to represent 
and manipulate zones and zone histories.
Unfortunately, DBM can only represent constraints corresponding to
\emph{convex} polyhedra (\ie{} clock constraints described by the grammar in Def.~\ref{def:ta}).
As a consequence, operations on DBM do not include union (or disjunctive constraints, respectively) which is, however, 
required for joining histories during the construction of composite zone-history graphs (see Def.~\ref*{def:history-join}).
Hence, for this particular step during the construction 
of composite zone-history graphs (\ie{} the last bullet point in Def.~\ref{def:composite-zone-history-graph}), 
we make use of an external call to an ILP-solver.
These additional calls may drastically impact the overall performance
of timed bisimilarity checking in case of non-deterministic TA as compared to deterministic
models (cf. Section~\ref{section:evaluation}).
These checks are conducted as follows:
Given a history $\mathcal{H}$ of TA $\mathcal{A}$ 
and a joint history $\mathcal{H}'$ of the respective composite zone-history graph, 
we have to check whether $\mathcal{H}$ is \emph{included} in $\mathcal{H}'$.
To this end, we consider an element-wise conjunction of the respective histories $\mathcal{H}$ and $\mathcal{H}'$, 
where we negate the elements of $\mathcal{H}'$, and then check the resulting 
conjunction for satisfiability.
For instance, if $H$ is the first element of $\mathcal{H}$ and $H'$ the first element of $\mathcal{H}'$,
we check if $H\land\neg H'$ is satisfiable.
If this is the case, then the behavior of $\mathcal{ZH}_\mathcal{A}$ is not yet 
completely included in $\mathcal{ZH}_{\mathcal{A}\times\mathcal{A}'}$, 
and we need to add the respective transition of $\mathcal{ZH}_\mathcal{A}$ to $\mathcal{ZH}_{\mathcal{A}\times\mathcal{A}'}$.
However, if $H\land\neg H'$ is not satisfiable, then the behavior of $\mathcal{ZH}_\mathcal{A}$ is already included 
in $\mathcal{ZH}_{\mathcal{A}\times\mathcal{A}'}$.
In our implementation, we utilize IBM ILOG CPLEX for these checks~\cite{CPLEX2015}.

In contrast to the theoretical constructions described in the previous section,
\timbrcheck{} is obviously not able to first construct a (potentially) infinite zone-history graph
before applying Algorithm~\ref{algorithm:finite-zhg} for pruning it to a finite zone-history graph.
Instead, we incrementally interleave Algorithm~\ref{algorithm:finite-zhg} with zone-history graph construction
in order to perform on-the-fly pruning.
To this end, we apply the check in line~\ref*{line:if-start} whenever a new state is potentially
added to the zone-history graph.

Our tool implementation can be used to conduct experimental timed bisimilarity checking
using different bound values $b$ as will be described in the next section.

%% file: sections/evaluation.tex

\section{Experimental Evaluation}\label{section:evaluation}

In this section, we present experimental results 
gained from applying our tool implementation (see Section~\ref{section:implementation}) 
of the previously presented technique to a collection of TA models.
In particular, we consider the following research questions.

\paragraph{Research Questions.}
Our tool \timbrcheck{} allows us to investigate
the impact of parameter $b$ (see Sect.~\ref{section:timed-bisimulation}) 
on efficiency and precision of timed-bisimilarity checking.
Intuitively, we expect that increasing the value of $b$ has 
a negative impact on performance, but a positive impact on precision.
We expect that there exists a value for $b$
yielding the best trade-off between both criteria on average.
In contrast, as our approach does only potentially yield false positives but no false negatives,
we do not have to investigate recall measures (see Theorem~\ref{theorem:bzhg-bisim}).

In addition, we expect the presence of non-deterministic behavior
in input models to (negatively) impact performance of timed-bisimilarity checking
as compared to the deterministic case, 
due to the additional effort caused by the composite zone-history graph construction (see Def.~\ref{def:composite-zone-history-graph}).
In contrast, we expect that the presence or absence of non-determinism 
does, in contrast to the value of $b$, not directly impact precision.
To summarize, we consider the following research questions.

\begin{itemize}
  \item \textbf{RQ1 (Efficiency).} How does the value of $b$ 
  as well as the presence/absence of \emph{non-deterministic} behavior 
  impact \emph{computational effort} of timed-bisimilarity checking?

  \item \textbf{RQ2 (Precision).} How does the value of $b$ 
  as well as the presence/absence of \emph{non-deterministic} behavior
  impact \emph{precision} of timed-bisimilarity checking?

  \item \textbf{RQ3 (Trade-off).} Which value for $b$ constitutes, on average,
  the best \emph{efficiency/pre\-cision trade-off} for timed-bisimilarity checking?
\end{itemize}

\paragraph{Methods and Experimental Design.}
For systematically investigating and comparing the impact of 
different values of $b$, we execute the experimental runs with 
ten different instantiations of parameter $b$, namely 0, 1, 2, 3, 4, 5, 10, 20, 25, and 30.
As our baselines, we consider two cases:
\begin{itemize}
  \item $b=0$ (tracking no history information) is supposed to constitute
the most efficient, yet less precise instantiation, whereas
 \item $b=\infty$ (tracking history information of unbounded length) 
guarantees precise results, but presumably causes the highest computational effort.
\end{itemize}
To keep overall runtime of experiments realistic, we enforce a time-out of 30 minutes for 
checking timed bisimilarity, thus potentially leading to no final results 
for particular combinations of subject systems and values of $b$.
In addition, to keep the overall number of experimental results comprehensible, 
we only consider \emph{strong} bisimilarity-checking for scenarios without 
internal behavior and \emph{weak} bisimilarity-checking, otherwise.

\paragraph{Subject Systems.}
We consider five different TA models taken from 
community benchmarks, frequently being used in recent experimental evaluation of TA analysis techniques:
\begin{itemize}
  \item Train-Gate-Controller (TGC)~\cite{Alur1993}: railroad gate controller for a simple level crossing.
  \item Gear Controller (GC)~\cite{Lindahl2001}: component of the control system operating in a modern vehicle.
  \item Collision Avoidance (CA)~\cite{Jensen1996}: protocol for communication among users using an Ethernet-like medium.
  \item Root Contention Protocol (RCP)~\cite{Collomb2001}: IEEE 1394 root contention protocol of the FireWire bus.
  \item Audio/Video Components (AVC)~\cite{Havelund1997}: messaging protocol for communication between AV components.
\end{itemize}

Unfortunately, none of the community benchmarks we found originally includes any
non-determinism or $\tau$-steps.
Hence, in order to also investigate the impact of the presence of 
non-determinism and silent moves in our evaluation, we manually adapted these five models 
by sporadically adding non-deterministic choices as well as $\tau$-steps.
Overall, this results in 10 TA models, of which 5 models are deterministic 
and 5 models include non-determinism and $\tau$-steps.
Table~\ref{table:subject-systems} provides an overview of key properties of the considered models, 
including the number of locations, switches and clocks and the number of (syntactic) 
occurrences of clock resets within switch guards.
\begin{table}[tp]
  \centering
  \caption{Subject Systems}\label{table:subject-systems}
  \small
  \input{figures/table-case-studies.tex}
\end{table}
Here, numbers within brackets denote the properties of the adapted non-deterministic variants of the models
for those cases where the respective property differs from the original model.
Based on these original models, we consider two experimental settings
for executing timed-bisimilarity checking.

\begin{enumerate}
  \item We simply \emph{copy} the model and perform timed-bisimilarity checks between
the original model and its one-to-one copy (which should therefore succeed).
  \item We further \emph{mutate} the copied model to obtain a rich 
  corpus of similar, yet slightly differing models and perform bisimilarity checks between
the original model and its mutations (which may either succeed or fail).
\end{enumerate}

For the second setting, we employ an existing framework providing canonical mutation operators for TA~\cite{Aichernig2014}.
In contrast to classical mutation testing which is used for evaluating \emph{effectiveness} of testing techniques or test suites,
equivalent mutants are not problematic in our setting, but even desirable to investigate
efficiency and precision for both negative as well as positive cases.
We therefore selected two operators presumably having the highest probability to 
produce slightly different, yet similar mutants, namely:
\begin{itemize}
  \item operator \emph{invert resets} flips the reset set $R$ of a switch (\ie{} $R$ becomes $C\setminus R$) and
  \item operator \emph{change guards} changes a comparison operator in a guard of a switch (\eg{} $\leq$ becomes $\geq$).
\end{itemize}

We exhaustively applied both operators to all 10 subject systems.
From the resulting overall number of 268 mutants, 76 are equivalent 
(\wrt{} timed bisimilarity) to the original model (see Table~\ref{table:subject-systems}).
Our evaluation comprises an overall number of 2029 runs of \timbrcheck{} 
of which 512 should be (true) positives (including the 5 identical copies) and 1517 should be (true) negatives
in case of optimally precise results.
However, we do not have measurement results for every mutant and every value of $b$ 
due to our maximum time-out of 30 minutes.


\paragraph{Data Collection.}
To answer \textbf{RQ1}, we measure (1) CPU time and (2) memory consumption, aggregated over all mutants
of each subject system.
Concerning (1), we sum up the CPU times required for generating the (bounded) 
zone-history graphs and for subsequent bisimilarity checks.
According to Theorem~\ref{theorem:bzhg-bisim}, the result
of bounded timed-bisimilarity checking for a bound value $b<\infty$
may yield false positives, but no false negatives.
Hence, to answer \textbf{RQ2}, we only have to count the number of false positives.
We executed all experiments on an Intel Core i7-8700k machine with 6x3.7GHz, 4GB RAM and 
Windows 10. Our tool is implemented in Java using AdoptOpenJDK 11.0.6.10.

\paragraph{Results and Discussion.}
The measurement results for \textbf{RQ1} (efficiency) 
are shown in Fig.~\ref{fig:results-rq1}. 
The given values correspond to the sums of CPU times as described above.
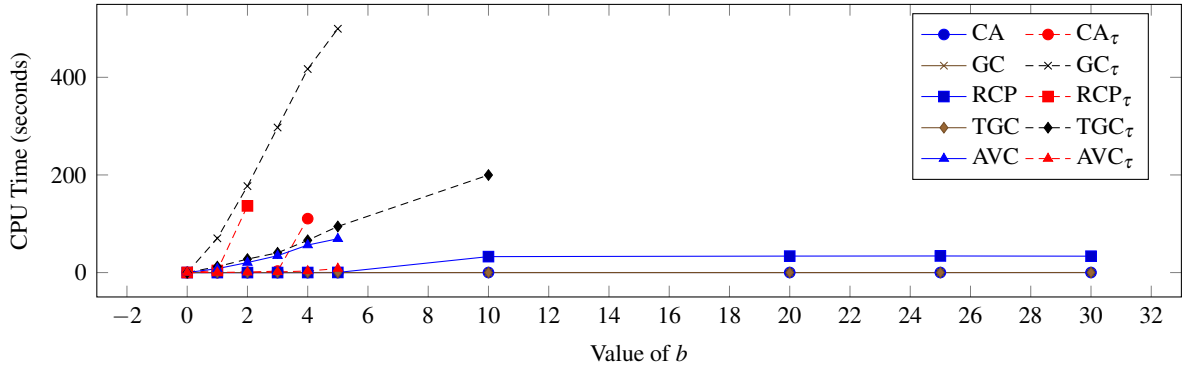
\begin{figure}[tp]
\centering
\input{figures/rq1.tex}
\caption{Measurement Results for \textbf{RQ1}}\label{fig:results-rq1}
\end{figure}
Non-deterministic subject systems with internal behavior are marked with index $\tau$.
As a first observation, the CPU time required for the timed-bisimilarity check
(having a peak value of 74ms, but in most cases performing much faster) 
is neglectable as compared to the CPU time required for the bounded zone-history
graph construction (ranging up to our time-out of 30 minutes).
Hence, we do not consider the CPU times independently but instead 
sum up the CPU times in Fig.~\ref{fig:results-rq1}.

For all deterministic subject systems except for RCP and AVC, the average CPU time is less than 200ms, 
whereas generating the zone-history graphs for RCP takes up to 34 seconds for $b\geq10$.
Furthermore, we already reached the time-out of 30 minutes for AVC for $b=10$.
In contrast, the computational effort for the non-deterministic subject systems heavily 
increases with increasing values of $b$.
As a result, we were only able to check these subject systems for smaller 
values of $b$ (ranging from $b\leq2$ for RCP$_\tau$ to $b\leq10$ for TGC$_\tau$).
This can be explained through the additional computational effort 
for generating \emph{composite} zone-history graphs.
We observe very similar tendencies for the memory consumption, ranging from 40MB to 200MB for 
deterministic systems, and going up to more than 1GB for non-deterministic systems (which we omitted in Fig.~\ref{fig:results-rq1}).

To summarize, \timbrcheck{} performs quite well for deterministic systems, whereas 
the results for non-deterministic systems indicate a worst-case exponential growth of the overall computational effort
(which is, however, inherent to the underlying theoretical problem).

The measurements for \textbf{RQ2} (precision) are shown in Fig.~\ref{fig:results-rq2}.
\begin{figure}[tp]
\centering
\input{figures/rq2.tex}
\caption{Measurement Results for \textbf{RQ2}}\label{fig:results-rq2}
\end{figure}
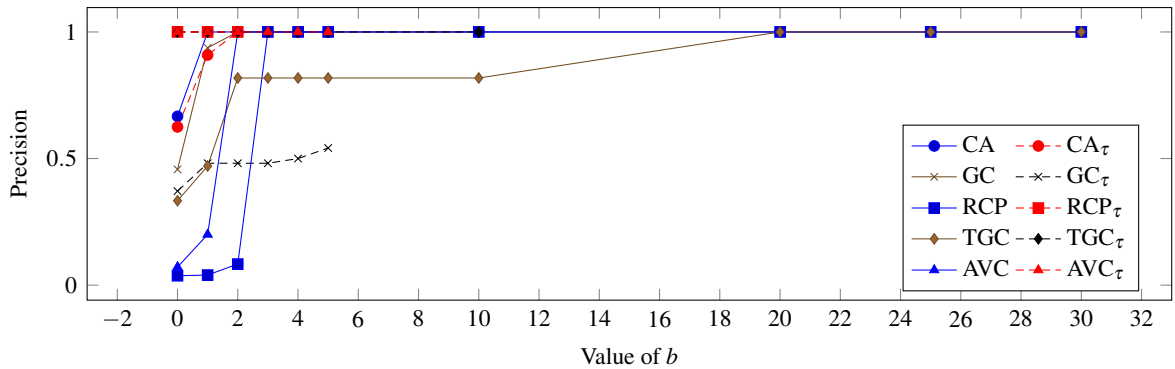
Furthermore, the box-plots in Fig.~\ref{fig:results-rq2-box} illustrate 
statistical distributions of the precision for each value of $b$.
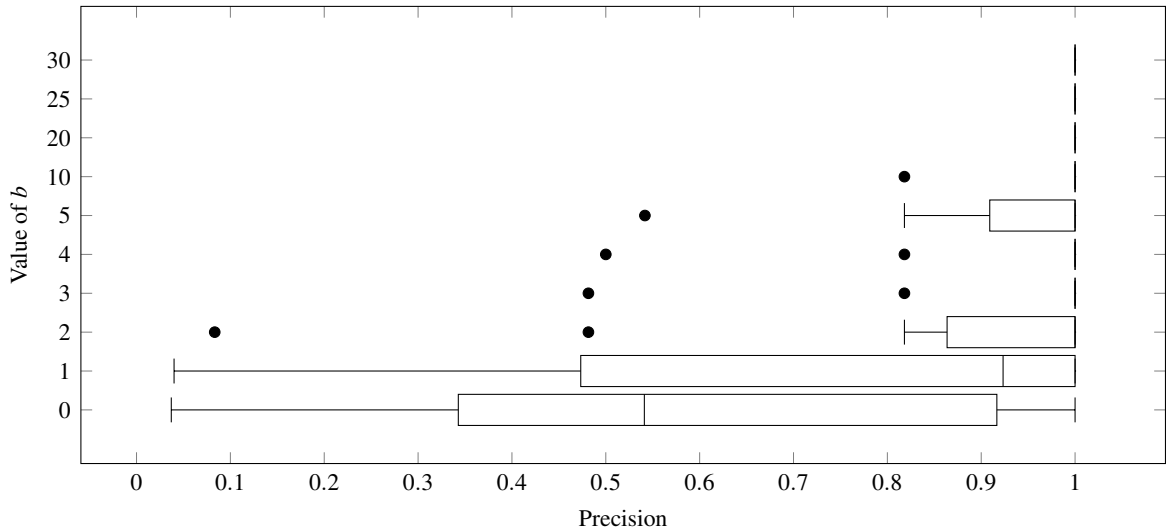
\begin{figure}[tp]
\centering
\input{figures/rq2-box-plot.tex}
\caption{Summary of Results for \textbf{RQ2}}\label{fig:results-rq2-box}
\end{figure}
Here, precision ranges from 0 to 1 and denotes the ratio of true positive 
results to the overall number of positive results.
Hence, a higher number of false positives (\ie{} non-bisimilar TA are reported as bisimilar) 
results in lower precision.
Interestingly, the median value for $b=0$ is 0.54, thus showing the essential 
necessity of including zone histories into timed-bisimilarity checks, even in case of 
smaller models with only one clock.
Conversely, we observe that from $b=3$ upwards, the probability of false positives 
drastically decreases with the interquartile range starting at 1 (while having 2 outliers).
For $b\geq10$, we observe no more false positives (except for one outlier for $b=10$).
Note, that the probability for false positives seems to increase for $b=5$.
However, we used a time-out of 30 minutes such that the box plot for $b=5$ actually comprises less subject systems
than in case of smaller values of $b$.
Furthermore, in case of GC$_\tau$, there is no value for $b$ without false positives
for which the timed bisimilarity check terminates before reaching the time-out.
In contrast, as expected, the presence/absence of non-determinism does not have 
a direct impact on precision.

Finally, based on these results, we
can conclude for \textbf{RQ3} (trade-off) that $b=3$ appears
to be a reasonable bound value for efficient, yet sufficiently
precise timed-bisimilarity checking regarding our subject systems.

\paragraph{Threats to Validity.}
We first discuss \emph{internal} threats.
The scope of our experimental setting is limited to the class of safety TA.
However, any non-trivial TA extension~\cite{Waez2013} obstructs essential properties of the underlying zone graphs, 
obviously making our approach more imprecise or even inapplicable.
Concerning the usage of mutation operators to synthetically generate
variations of our subject systems, we rely on small and locally restricted changes as usual.
Nevertheless, our experiments show that those mutations may produce both
TA which are equivalent to the original TA as well as TA which are not, thus
indicating mutation to be an appropriate tool for our experiments.
Finally, to ensure correctness of (a)~our theory and (b)~our tool implementation, we (a)~provide 
correctness proofs and (b)~exhaustively tested our tool on a rich collection of test cases in 
terms of particularly sophisticated pairs of TA fragments 
(which are also available on our accompanying web page\footnote{\url{https://www.es.tu-darmstadt.de/timbrcheck/}}).

We identified as \emph{external} threats (a)~a lack
of comparison to other tools and (b)~the relatively small
set of subject systems.
Concerning (a), there currently exists, to the best of our knowledge,
no competitive tool that provides a functionality being comparable to \timbrcheck{}.
Concerning (b), we selected our set of subject systems from
well-established community benchmarks of reasonable size and complexity which are frequently used
in experiments involving analysis techniques for TA.
However, we plan in a future work
to consider further case studies, especially including real-world systems.

%% file: figures/table-case-studies.tex

\begin{adjustbox}{max width=\textwidth}
\begin{tabular}{rccccc}
\toprule
& \textbf{TGC} & \textbf{GC} & \textbf{CA} & \textbf{RCP} & \textbf{AVC} \\	
\midrule
\textbf{\# Locations} & 14 (15) & 23 (24) & 6 (7) & 10 & 18 (19) \\
\textbf{\# Switches} & 18 (20) & 28 (32) & 13 (15) & 26 (28) & 30 (33) \\
\textbf{\# Clocks} & 1 & 1 & 1 & 2 & 1 \\
\textbf{\# Resets} & 6 & 12 & 1 (2) & 9 & 18 (19) \\
\textbf{\# Mutants} & 26 & 34 & 15 & 26 & 32 \\
\textbf{\# Bisimilar Mutants} & 11 & 15 & 9 & 2 & 1 \\
\textbf{\# Internal Transitions} & 0 (1) & 0 (1) & 0 (1) & 0 (1) & 0 (1) \\
\textbf{\# Non-det. Choices} & 0 (1) & 0 (3) & 0 (1) & 0 (1) & 0 (2) \\
\bottomrule
\end{tabular}
\end{adjustbox}

%% file: figures/rq1.tex
\begin{tikzpicture}
\begin{axis}
[
width=\linewidth,
height=.25\textheight,
yticklabel style={
	/pgf/number format/fixed,
	/pgf/number format/precision=2
},
xlabel = Value of $b$,
ylabel = CPU Time (seconds),
scaled y ticks=false,
legend cell align={left},
legend pos=north east,
legend columns=2,
]

\addplot+[style=solid, mark=*] coordinates {
(0,0.0330625)
(1,0.0253125)
(2,0.0016875)
(3,0.002375)
(4,0.004125)
(5,0.007)
(10,0.1411875)
(20,0.118625)
(25,0.1205)
(30,0.1130625)
};

\addplot+[style=densely dashed, mark=*, mark options={solid}] coordinates {
(0,0.0314375)
(1,0.1866875)
(2,0.6403125)
(3,3.1093125)
(4,110.4865625)
};

\addplot+[style=solid, mark=x] coordinates {
(0,0.000742857)
(1,0.000485714)
(2,0.001)
(3,0.001085714)
(4,0.001257143)
(5,0.001342857)
(10,0.003228571)
(20,0.004028571)
(25,0.004057143)
(30,0.003942857)
};

\addplot+[style=densely dashed, mark=x, mark options={solid}] coordinates {
(0,0.003314286)
(1,70.20482857)
(2,177.1211143)
(3,297.1722286)
(4,416.9733714)
(5,499.4306286)
};

\addplot+[style=solid, mark=square*] coordinates {
(0,0.005185185)
(1,0.005037037)
(2,0.012074074)
(3,0.037333333)
(4,0.123037037)
(5,0.296444444)
(10,32.70085185)
(20,33.7767037)
(25,34.01866667)
(30,33.59222222)
};

\addplot+[style=densely dashed, mark=square*, mark options={solid}] coordinates {
(0,0.206740741)
(1,4.283592593)
(2,136.5603333)
};

\addplot+[style=solid, mark=diamond*] coordinates {
(0,0.000555556)
(1,0.016518519)
(2,0.000666667)
(3,0.000703704)
(4,0.001296296)
(5,0.001333333)
(10,0.004962963)
(20,0.014962963)
(25,0.041925926)
(30,0.171740741)
};

\addplot+[style=densely dashed, mark=diamond*, mark options={solid}] coordinates {
(0,0.003185185)
(1,12.13718519)
(2,27.63692593)
(3,40.45862963)
(4,66.181)
(5,94.46714815)
(10,199.81468)
};

\addplot+[style=solid, mark=triangle*] coordinates {
(0,0.04730303)
(1,8.093121212)
(2,20.49230303)
(3,34.57163636)
(4,56.40612121)
(5,69.28127273)
};

\addplot+[style=densely dashed, mark=triangle*, mark options={solid}] coordinates {
(0,0.039)
(1,0.219333333)
(2,0.596878788)
(3,1.348818182)
(4,3.23030303)
(5,8.198909091)
};

\legend{CA, CA$_\tau$, GC, GC$_\tau$, RCP, RCP$_\tau$, TGC, TGC$_\tau$, AVC, AVC$_\tau$}
\end{axis}
\end{tikzpicture}

%% file: figures/rq2.tex
\begin{tikzpicture}
\begin{axis}
[
width=\linewidth,
height=.25\textheight,
yticklabel style={
	/pgf/number format/fixed,
	/pgf/number format/precision=2
},
xlabel = Value of $b$,
ylabel = Precision,
scaled y ticks=false,
legend cell align={left},
legend pos=south east,
legend columns=2,
]

\addplot+[style=solid, mark=*] coordinates {
(0,0.666666667)
(1,1)
(2,1)
(3,1)
(4,1)
(5,1)
(10,1)
(20,1)
(25,1)
(30,1)
};

\addplot+[style=densely dashed, mark=*, mark options={solid}] coordinates {
(0,0.625)
(1,0.909090909)
(2,1)
(3,1)
(4,1)
};

\addplot+[style=solid, mark=x] coordinates {
(0,0.457142857)
(1,0.9375)
(2,1)
(3,1)
(4,1)
(5,1)
(10,1)
(20,1)
(25,1)
(30,1)
};

\addplot+[style=densely dashed, mark=x, mark options={solid}] coordinates {
(0,0.371428571)
(1,0.481481481)
(2,0.481481481)
(3,0.481481481)
(4,0.5)
(5,0.541666667)
};

\addplot+[style=solid, mark=square*] coordinates {
(0,0.037037037)
(1,0.04)
(2,0.083333333)
(3,1)
(4,1)
(5,1)
(10,1)
(20,1)
(25,1)
(30,1)
};

\addplot+[style=densely dashed, mark=square*, mark options={solid}] coordinates {
(0,1)
(1,1)
(2,1)
};

\addplot+[style=solid, mark=diamond*] coordinates {
(0,0.333333333)
(1,0.470588235)
(2,0.818181818)
(3,0.818181818)
(4,0.818181818)
(5,0.818181818)
(10,0.818181818)
(20,1)
(25,1)
(30,1)
};

\addplot+[style=densely dashed, mark=diamond*, mark options={solid}] coordinates {
(0,1)
(1,1)
(2,1)
(3,1)
(4,1)
(5,1)
(10,1)
};

\addplot+[style=solid, mark=triangle*] coordinates {
(0,0.071428571)
(1,0.2)
(2,1)
(3,1)
(4,1)
(5,1)
};

\addplot+[style=densely dashed, mark=triangle*, mark options={solid}] coordinates {
(0,1)
(1,1)
(2,1)
(3,1)
(4,1)
(5,1)
};

\legend{CA, CA$_\tau$, GC, GC$_\tau$, RCP, RCP$_\tau$, TGC, TGC$_\tau$, AVC, AVC$_\tau$}
\end{axis}
\end{tikzpicture}

%% file: figures/rq2-box-plot.tex

\begin{tikzpicture}
\begin{axis}[
width=\linewidth,
height=.35\textheight,
xlabel={Precision},
ytick={1,2,3,4,5,6,7,8,9,10},
yticklabels={0,1,2,3,4,5,10,20,25,30},
ylabel={Value of $b$},
]

\addplot+[boxplot, style=solid, mark=*, color=black, mark options={color=black}]
table [row sep=\\,y index=0] {
0.037037037\\ 0.071428571\\ 0.333333333\\ 0.371428571\\ 0.457142857\\ 0.625\\ 0.666666667\\ 1\\ 1\\ 1\\
};

\addplot+[boxplot, style=solid, mark=*, color=black, mark options={color=black}]
table [row sep=\\,y index=0] {
0.04\\ 0.2\\ 0.470588235\\ 0.481481481\\ 0.909090909\\ 0.9375\\ 1\\ 1\\ 1\\ 1\\
};

\addplot+[boxplot, style=solid, mark=*, color=black, mark options={color=black}]
table [row sep=\\,y index=0] {
0.083333333\\ 0.481481481\\ 0.818181818\\ 1\\ 1\\ 1\\ 1\\ 1\\ 1\\ 1\\
};

\addplot+[boxplot, style=solid, mark=*, color=black, mark options={color=black}]
table [row sep=\\,y index=0] {
0.481481481\\ 0.818181818\\ 1\\ 1\\ 1\\ 1\\ 1\\ 1\\ 1\\
};

\addplot+[boxplot, style=solid, mark=*, color=black, mark options={color=black}]
table [row sep=\\,y index=0] {
0.5\\ 0.818181818\\ 1\\ 1\\ 1\\ 1\\ 1\\ 1\\ 1\\
};

\addplot+[boxplot, style=solid, mark=*, color=black, mark options={color=black}]
table [row sep=\\,y index=0] {
0.541666667\\ 0.818181818\\ 1\\ 1\\ 1\\ 1\\ 1\\
};

\addplot+[boxplot, style=solid, mark=*, color=black, mark options={color=black}]
table [row sep=\\,y index=0] {
0.818181818\\ 1\\ 1\\ 1\\ 1\\
};

\addplot+[boxplot, style=solid, mark=*, color=black, mark options={color=black}]
table [row sep=\\,y index=0] {
1\\ 1\\ 1\\ 1\\
};

\addplot+[boxplot, style=solid, mark=*, color=black, mark options={color=black}]
table [row sep=\\,y index=0] {
1\\ 1\\ 1\\ 1\\
};

\addplot+[boxplot, style=solid, mark=*, color=black, mark options={color=black}]
table [row sep=\\,y index=0] {
1\\ 1\\ 1\\ 1\\
};

\end{axis}
\end{tikzpicture}

%% file: sections/conclusion.tex

\section{Conclusion}\label{section:conclusion}

We presented a novel formalism, called
bounded zone-history graphs, for precise, yet scalable timed-bisimilarity checking of
non-deterministic TA with silent moves.
Our tool \timbrcheck{} currently supports checking strong bisimilarity 
as well as weak bisimilarity for deterministic and non-deterministic TA provided in the \uppaal{} file format.
Our experimental evaluation shows promising potential in scaling bisimilarity 
checking for deterministic TA also to larger-scaled models without seriously harming precision.
As a future work, we plan to extend our tool and our accompanying experiments to more advanced classes of TA~\cite{Waez2013}.
In addition, we are interested in adapting our technique
to incorporate further crucial notions of behavioral equivalences beyond
timed bisimulation.